\documentclass[a4paper,UKenglish,cleveref,autoref, thm-restate, numberwithinsect]{lipics-v2021}

\usepackage{amsthm}
\usepackage{amssymb}
\usepackage{amsmath}
\usepackage{mathtools}
\usepackage{csquotes}
\usepackage{enumitem}
\usepackage{svg}
\usepackage{hyperref}
\usepackage{xspace}

\newcommand{\defn}[1]{{\textit{\textbf{\boldmath #1}}}\xspace}
\renewcommand{\paragraph}[1]{\vspace{0.09in}\noindent{\bf \boldmath #1.}}

\newcommand{\LH}{\mathsf{LH}}
\newcommand{\ZH}{\Z\mathsf{H}}
\newcommand{\FH}{\F\mathsf{H}}
\newcommand{\RH}{\R\mathsf{H}}

\newcommand{\MLH}{\mathsf{2LH}}
\newcommand{\SLH}{\mathsf{SH}}
\newcommand{\randH}{\mathcal{R}\mathsf{H}}

\newcommand{\blockFp}{\square\F_p\mathsf{H}}
\newcommand{\strideFp}{\bigcirc\F_p\mathsf{H}}
\newcommand{\blockZm}{\square\Z_m\mathsf{H}}
\newcommand{\strideZm}{\bigcirc\Z_m\mathsf{H}}

\newcommand{\lap}{\mathfrak{e}}
\newcommand{\halfp}{[\ceil{p/2}]_{\setminus 0}}
\newcommand{\pnozero}{[p]_{\setminus 0}}

\newcommand{\circabs}{\mathfrak{s}}

\newcommand{\modp}{\mathfrak{m}_p}
\newcommand{\modm}{\mathfrak{m}_m}
\newcommand{\posmod}{\mathfrak{m}}

\newcommand{\bigO}{\mathcal{O}}
\newcommand{\tilo}{\widetilde{\bigO}}
\DeclareMathOperator{\E}{\mathbb{E}}

\DeclareMathOperator{\sgn}{sgn}

\newcommand{\PRM}{\mathbb{P}}

\DeclareMathOperator{\poly}{\text{poly}}

\newcommand{\interior}[1]{ {\kern0pt#1}^{\mathrm{o}} }

\newcommand{\eps}{\varepsilon}
\renewcommand{\d}{\mathrm{d}} 

\newcommand{\setof}[2]{\left\{ #1\; : \;#2 \right\}}
\newcommand{\set}[1]{\left\{ #1\right\}}

\newcommand{\R}{\mathbb{R}}
\newcommand{\Q}{\mathbb{Q}}

\newcommand{\Z}{\mathbb{Z}}
\newcommand{\F}{\mathbb{F}}

\newcommand{\N}{\mathbb{N}}

\usepackage{algorithm}
\usepackage[noend]{algpseudocode} 

\newcommand{\floor}[1]{\left\lfloor #1 \right\rfloor}
\newcommand{\dfloor}[1]{\lfloor #1 \rfloor}
\newcommand{\ceil}[1]{\left\lceil #1 \right\rceil}
\newcommand{\paren}[1]{\left( #1 \right)}

\newcommand{\abs}[1]{\left| #1 \right|}

\usepackage{cleveref}
\crefname{equation}{}{} 

\newtheorem{innercustomthm}{}
\newenvironment{customthm}[1]
  {\renewcommand\theinnercustomthm{#1}\innercustomthm}
  {\endinnercustomthm}
\newtheorem{conj}[theorem]{Conjecture}
\newtheorem{question}[theorem]{Question}
\newtheorem{prop}[theorem]{Proposition}
\newtheorem{cor}[theorem]{Corollary}
\newtheorem{fact}[theorem]{Fact}

\theoremstyle{definition}
\newtheorem{defin}[theorem]{Definition}
\newtheorem{rmk}[theorem]{Remark}

\pdfoutput=1 
\hideLIPIcs  
\nolinenumbers

\author{Alek Westover}{MIT, USA}{alekw@mit.edu}{}{}
\title{On the Relationship Between Several Variants of the Linear Hashing Conjecture}
\authorrunning{Alek Westover} 

\begin{document}

\maketitle

\abstract{
  In \defn{Linear Hashing} ($\LH$) with $n$ bins on a size
  $u$ universe ${\mathcal{U}=\set{0,1,\ldots, u-1}}$, items
  $\set{x_1,\ldots, x_n}\subset \mathcal{U}$ are placed in
  bins by the hash function
  $x_i\mapsto (ax_i+b)\bmod p \bmod n$
  for some prime $p\in [u,2u]$ and randomly chosen
  integers $a,b \in [1,p]$. The \defn{maxload} of $\LH$ is the
  number of items assigned to the fullest bin. Expected maxload for
  a worst-case set of items is a natural measure of how well $\LH$
  distributes items amongst the bins.

  Despite $\LH$'s simplicity, bounding $\LH$'s worst-case
  maxload is extremely challenging.
  It is well-known that on random inputs $\LH$ achieves
  maxload $\Omega\left(\frac{\log n}{\log\log n}\right)$; this is
  currently the best lower bound for $\LH$'s expected maxload.
  Recently Knudsen established an upper bound of
  $\widetilde{\bigO}(n^{1 / 3})$.
  The question ``Is the worst-case expected maxload of $\LH$
  $n^{o(1)}$?'' is one of the most basic open problems in discrete
  math.

  In this paper we propose a set of intermediate open questions to
  help researchers make progress on this problem.
  We establish the relationship between these intermediate open
  questions and make some partial progress on them.
}

\section{Introduction}
The \defn{hashing problem} is to assign $n$ items from a
universe $\mathcal{U}$ to $\beta$ \defn{bins} so that all bins receive a
similar number of items. In particular, we measure the quality
of a hashing scheme's load distribution by the number of items
in the fullest bin; we refer to this quantity as the \defn{maxload}.
We desire three main properties of a hashing scheme:
(1) small expected maxload for all sets of items, (2) fast
evaluation time, and (3) small description size.

\paragraph{Related Work}
The hashing problem has been extensively studied.
We use the standard parameters $|\mathcal{U}|\in \poly(n)$,
$\beta=n$ in the following discussion. We also assume the
unit-cost RAM model, i.e., that arithmetic operations on
numbers of size $\Theta(\log n)$ can be performed in constant
time. We use the abbreviation $\ell(n) = \frac{\log n}{\log\log n}$.

The hash function which assigns each item independently
randomly to a bin achieves the smallest possible expected
maxload in general, namely $\Theta(\ell(n))$
\cite{mitzenmacher2017probability}.
However, describing a fully random function requires
$\Omega(|\mathcal{U}|\log n)$ bits which is extremely large.
Full independence is not necessary to achieve optimal maxload.
For instance, Carter and Wegman \cite{carter1977universal} show
that degree $\Theta(\ell(n))$ polynomials over a finite
field constitute a $\Theta(\ell(n))$-wise
independent hash family while still achieving maxload
$\Theta(\ell(n))$.
Improving on this result, Celis et al. \cite{celis2013balls}
demonstrate a hash family achieving maxload $\Theta(\ell(n))$
with evaluation time $\bigO(\sqrt{\log n})$.

In fact, it is even possible to achieve optimal maxload with
constant evaluation time, as demonstrated by Siegel in
\cite{siegel_universal_2004}. However, Siegel's hash function
has description size $\poly(n)$.
Furthermore, Siegel proved that it is impossible to
simultaneously achieve optimal maxload, constant evaluation
time, $n^{o(1)}$ description size, and $\Omega(\ell(n))$ independence.
However, this still leaves room for improvement if we do not
require large degrees of independence.

By itself small independence does not give any good bound on
maxload; for example, there are pairwise-independent hash
families with maxload $\Omega\left(\sqrt{n}\right)$ \cite{petershor}.
However, Alon et al. \cite{alon_is_1997} show that the
pairwise-independent hash family of multiplication by random
matrices over $\F_2$ achieves $\bigO(\ell(n) \cdot (\log \log
    n)^2)$ maxload in $\bigO(\log(n))$ evaluation time.

The following question remains open:
\begin{question}
  \label{question:fastandgood}
  Is there a hash family with $\bigO(1)$ machine
  word description whose evaluation requires $\bigO(1)$ arithmetic
  operations that has expected maxload bounded by $n^{o(1)}$?
\end{question}

\paragraph{Linear Hashing ($\LH$)}
$\LH$ \cite{motwani1995randomized, cormen2022introduction,
  sedgewick2014algorithms} is an attractive potential solution to
\cref{question:fastandgood}, trivially satisfying the conditions
of small description size and fast evaluation time.
Despite $\LH$'s simplicity, understanding its
maxload is a notoriously challenging and wide open question.
The best known lower bound on $\LH$'s maxload is $\Omega\left(\frac{\log
    n}{\log\log n}\right)$, whereas the best known upper bound is
$\widetilde{\bigO}(n^{1/3})$ due to an elegant combinatorial
argument of Knudsen \cite{knudsen_linear_2017}.

Let $\mathcal{U} = \set{0,1,\ldots, u-1}$ denote the universe.
To keep the introduction simple we abuse notation and let $x
  \bmod m$ denote the unique representative of the equivalence
class $x +m\Z$ lying in $[0,m)$.
In most textbooks (e.g., \cite{motwani1995randomized}) $\LH$ is
defined placing $x\in \mathcal{U}$ in bin
\[(ax+b)\mod p \mod \beta\]
for prime $p\in [u, 2u]$
and randomly chosen integers $a,b \in [1,p]$;
we refer to this as \defn{strided hashing}.
In \cite{dietzfelbinger1997reliable} Dietzfelbinger et al. give an
alternative definition placing $x\in \mathcal{U}$ in bin
\[  \floor{ \frac{(ax+b)\mod p}{p/\beta} };\]
we refer to this as \defn{blocked hashing}.

The \defn{maxload} of a hash family with respect to an $n$-element set
$X\subset\mathcal{U}$ is a random variable counting the number of
$x\in X$ hashed to the fullest bin under a randomly chosen hash
function. We aim to minimize the expected maxload for worst-case
$X$.

For $\beta=n$ Knudsen \cite{knudsen_linear_2017} implicitly observed that
the maxload of blocked and strided hashing differ by
at most a factor-of-$2$; this follows from our
\cref{prop:blockZisok}. Roughly this equivalence follows by observing that
if blocked hashing has large maxload for $a=a_0,$ then strided
hashing will have large maxload for $a=a_0 n \bmod p$. Similarly,
if strided hashing has large maxload for $a=a_1$ then blocked
hashing will have large maxload for $a=a_1 n^{-1}\bmod p$, where
$n^{-1}$ is the multiplicative inverse of $n$ in $\F_p$.
Thus, classically blocked and strided hashing are
essentially equivalent. On the other hand we show that blocked hashing
generalizes more readily.\footnote{In \cref{prop:blockZsucks} we
  show that strided hashing does not generalize cleanly to composite
  moduli. Furthermore, there is no natural way to generalize
  strided hashing to real numbers.}
Thus, the majority of our results will concern blocked hashing.
We further simplify blocked hashing by
removing the $+b$ \defn{``shift term''} obtaining
\begin{equation} \label{eq:LHdefn3}
  x\mapsto \floor{\frac{ax \bmod p}{p/\beta}}.
\end{equation}
Removing the shift term also will not impact the maxload by more
than a factor-of-$2$: changing the shift term at most splits
fullest bins in half or merges parts of adjacent bins
into a single new fullest bin.
For the rest of this section \defn{Simple $\LH$} will refer to
the hashing scheme defined in \cref{eq:LHdefn3}.
We propose \cref{q:isLHsubpoly} as a potential solution to
\cref{question:fastandgood}.
\begin{question}\label{q:isLHsubpoly}
  Is the worst-case expected maxload of Simple $\LH$ bounded by $n^{o(1)}$?
\end{question}


\subsection{Our Results}
In this paper we propose a set of intermediate open questions to
help researchers make progress on \cref{q:isLHsubpoly}.
We establish the relationship between these intermediate open
questions and make some partial progress on them.
Except for in \cref{sec:twobins} we take $\beta=n$ bins.

\paragraph{Connecting Prime and Integer Moduli}
In \cref{sec:Z} we consider the importance of
using a prime modulus for $\LH$.
Conventional wisdom (e.g., \cite{cormen2022introduction}) is that
using a non-prime modulus is catastrophic. Using a
non-prime modulus is complicated by the fact that in a general
ring, as opposed to a finite field, non-zero elements can
multiply to zero.
Fortunately for any $m$ there is a reasonably large subset of
$\Z_m$ which forms a group under multiplication. The subset is
$\Z_m^{\times }$: the set of integers in $\Z_m$ coprime to $m$.
We define an alternative version of $\LH$ called \defn{Smart} $\LH$
where $a$ is chosen uniformly from $\Z_m^{\times}$ rather than
$\Z_m$. We show:

\begin{customthm}{\cref{thm:LHSLH}}
  Fix integer $m\in \poly(n)$. The expected maxloads of
  Smart $\LH$ with modulus $m$ and Simple $\LH$ with modulus $m$
  differ by at most a factor-of-$n^{o(1)}$.
\end{customthm}
Intuitively, Smart $\LH$ with composite modulus behaves somewhat
similarly to Simple $\LH$ with prime modulus.
This similarity allows us to, with several new ideas, translate Knudsen's proof
\cite{knudsen_linear_2017} of a $\widetilde{\bigO}(n^{1/3})$
bound on Simple $\LH$'s maxload for prime modulus to the
composite modulus setting, giving:
\begin{customthm}{\cref{thm:Zm1_3}}
  The expected maxload of Smart $\LH$ is at most
  $\widetilde{\bigO}(n^{1/3})$.
\end{customthm}

Part of why \cref{thm:Zm1_3} is interesting is that using
\cref{thm:Zm1_3} in \cref{thm:LHSLH} gives:
\begin{customthm}{\cref{cor:translate}}
  The expected maxload of Simple $\LH$ with composite modulus is
  at most $n^{1/3+o(1)}$.
\end{customthm}
In particular, in \cref{cor:translate} we have translated the
state-of-the-art bound for maxload from the prime modulus setting
to the composite modulus setting.
This gives tentative evidence that the behavior of $\LH$ with
composite modulus may actually be the same as that of $\LH$ with
prime modulus. We leave this as an open question:
\begin{customthm}{\cref{question:equivalenceFZ}}
  Is the worst-case maxload of composite modulus $\LH$ the
  same, up to a factor-of-$n^{o(1)}$, as that of prime
  modulus $\LH$?
\end{customthm}

\paragraph{Connecting Integer and Real Moduli}
In \cref{sec:R} we consider \defn{Real $\LH$} where the
multiplier ``$a$'' in \eqref{eq:LHdefn3} is chosen
from $\R$. Initially the change to a continuous setting seems
to produce a very different problem.
In this continuous setting one equivalent way of formulating
\cref{q:isLHsubpoly} is:
\begin{question}[``Crowded Runner Problem'']
  \label{question:dual}
  Say we have $n$ runners with distinct speeds $x_1,x_2, \ldots,
    x_n \in (0,1)$ starting at the same location on a length $1$
  circular race-track. $a\in (0,1)$ is chosen randomly and all
  runners run from time $0$ until time $a$.
  Is it true that on average the largest ``clump'' of runners,
  i.e., set of runners in single interval of size $1/n$, is of
  size at most $n^{o(1)}$?
\end{question}
As formulated in \cref{question:dual} the problem becomes a
dual to the famous unsolved ``Lonely Runner Conjecture''
of Wills \cite{wills1967zwei}  and Cusick \cite{cusick1982view}
as formulated in \cite{bienia1998flows}. In the Lonely Runner
Conjecture the question is for each runner whether there is any
time such that the runner is ``lonely'', i.e., separated from all
other runners by distance at least $1/n$. Our question is whether for
most time steps there is any runner that is ``crowded'',
i.e., with many other runners within an interval of size $1/n$
around the runner.
The difficulty of the Lonely Runner Conjecture may be indicative
that the ``Crowded Runner Conjecture'' is also quite difficult.

In \cref{thm:itisreal} we show a surprising equivalence between
$\LH$ for integer and real moduli. More specifically, in our
lower bound on Real $\LH$ we compare to a potentially stronger
version of integer modulus $\LH$ termed \defn{Random Modulus
$\LH$} where the modulus is not simply the universe size $u$, but
rather a randomly chosen (and likely composite) integer in $[u/2,
u]$. Random Modulus $\LH$ is clearly at most a factor-of-$2$
worse that Simple $\LH$, but it is not obvious whether it is any
better; we leave this as an open question:
\begin{question}
  Does Random Modulus $\LH$ achieve asymptotically smaller
  expected maxload than Simple $\LH$?
\end{question}

Formally the equivalence between Real $\LH$ and Random Modulus
$\LH$ can be stated as follows:
\begin{customthm}{\cref{thm:itisreal}}
  Let $f(n)$ be a lower bound on Random Modulus $\LH$'s expected
  maxload that holds for all sufficiently large universes.
  Let $g(n)$ be an upper bound on Simple $\LH$'s expected maxload
  that holds for all sufficiently large universes.
  Let $M_\R$ denote the expected maxload of
  Real $\LH$. Then
  \[
  \Omega\paren{\frac{f(n)}{\log\log n}} \le M_\R \le \bigO(g(n)).
\]
\end{customthm}
The proof of this theorem involves several beautiful
number-theoretical lemmas and is one of our main technical
contributions.

This equivalence between the real and integer versions of $\LH$ shows
that \cref{q:isLHsubpoly} may not fundamentally be about prime
numbers or even integers.


\paragraph{A Simpler Problem: The Two Bin Case}
Finally in \cref{sec:twobins} we consider an even simpler question than
\cref{q:isLHsubpoly}:
\begin{question}\label{question:twobincase}
  What can be said about Simple $\LH$ in the case where there are only
  $\beta=2$ bins?
\end{question}

Intuitively, Simple $\LH$ should place roughly half of the balls
in each bin. In fact, one might even conjecture (especially if we
believe that Simple $\LH$ does well on many bins) that Simple
$\LH$ should achieve a Chernoff-style concentration bound on the
number of balls in each of the two bins. What makes
\cref{question:twobincase} interesting is that even analyzing the
\emph{expected} maxload of Simple $\LH$ in the two-bin case is
already a nontrivial problem (because, unlike non-simple $\LH$,
Simple $\LH$ is not pairwise independent). In fact, it may be the
simplest non-trivial problem that one can state about Simple
$\LH$.

In \cref{sec:twobins} we prove a partial result towards
\cref{question:twobincase}. We show that, even though some pairs
of elements may have probability as large as $2/3$ of colliding
in their bin assignment, one can nonetheless establish a $(1 +
o(1)) n/2$ bound on the expected maxload:
\begin{customthm}{\cref{thm:dontneedb}}
  Simple $\LH$ with $\beta=2$ bins has expected maxload at most
  $n/2 + \widetilde{\bigO}(n^{1/2}).$
\end{customthm}

We propose proving stronger results in the two-bin setting as a
fruitful direction for future research. For example, establishing
Chernoff-style concentration bounds on the maxload in the two-bin
case would constitute the strongest evidence to date that Simple
$\LH$ is a good load-balancing function.

\section{Preliminaries}
\paragraph{Set Definitions}
We write $\PRM$ to denote the set of primes, $\F_p$ for $p\in \PRM$ to denote the
finite field with $p$ elements, and $\Z_m$ for $m\in \N$ to denote
the ring $\Z / m\Z.$

For $a,b\in \R$ we define
$[a,b]=\setof{x\in \R}{a\le x\le b}.$
For $n\in \N$ we define
$[n] = \set{0,1,\ldots, n-1}$, and 
$[n]_{\setminus 0} = [n]\setminus \set{0}.$
For $p\in \N,x\in \R$ we define $\modp(x)$ as
the unique number in the set
$(x+p\Z) \cap [0,p),$
and define
$\circabs_p(x)= \min(\modp(x), \modp(-x)).$
For $x\in \Z$, $\modp(x)$ is the positive remainder
obtained when $x$ is divided by $p$ and $\circabs_p(x)$ is the
smallest distance to an element of $p\Z$ from $x$.
However these functions are also defined for $x\in \R\setminus \Z$.

For $a\in \R$ and set $X\subset \R$ we define
$a+X = \setof{a+x}{x\in X},$
$a\cdot X = \setof{a\cdot x}{x\in X},$
$\modp(X) = \setof{\modp(x)}{x\in X}.$

\paragraph{Number Theoretic Definitions}
For $x,y\in \N$ we write $x\perp y$ to denote that  $x,y$ are
coprime, and $x\mid y$ to denote that $x$ divides $y$.
We write $\gcd(x,y)$ to denote the largest $k$ satisfying both  $k\mid
  x$  and $k\mid y$.
A \defn{unit}, with respect to an implicit ring, is an element
with an inverse. For $m\in \N$ we define $\Z_m^{\times}$ as the
set of units in $\Z_m$. That is,
$\Z_m^{\times } =\setof{k \in [m]}{k\perp m}.$ We write
$\phi$ to denote the Euler-Toitent function, which is defined to
be $\phi(m)=|\Z_m^{\times }|$. We write $\tau(m)$ to denote the number of
divisors of $m$.
The following two facts (see, e.g.,
\cite{hardy1979introduction}), will be
useful in several bounds:
\begin{fact}\label{fact:numdivs} $\tau(n) \le 2^{\bigO(\log n / \log\log n)} \leq n^{o(1)}.$
\end{fact}
\begin{fact}\label{fact:toitent}
  $\frac{n}{2\log\log n} \le \phi(n)< n.$
\end{fact}

\paragraph{Hashing Definitions}
A \defn{hashing scheme} mapping universe $[u]$ to $\beta$ bins is a set of functions
$\setof{h_i: [u]\to [\beta]}{i\in I}$
parameterized by $i\in I$ for some set $I$.
We say $h_i$ sends element  $x$ to bin $h_i(x)$.
The \defn{maxload} of $\setof{h_i}{i\in I}$ with respect to set $X$ for
parameter choice $i_0\in I$ is
\[
\max_{k\in [\beta]} \left|\setof{x\in X}{h_{i_0}(x)=k}\right|.
\]
In other words, maxload is the number of elements mapped to the
fullest bin.
We are concerned with bounding the expected maxload of hashing
schemes with respect to uniformly randomly
chosen parameter $i\in I$ for arbitrary $X$.
We will abbreviate uniformly randomly to randomly when the
uniformity is clear from context.

\begin{rmk}\label{rmk:assumesize}
  Our analysis is asymptotic as a function of $n$, the number of
  hashed items. We assume $n$ is at least a sufficiently large constant.
  We also require the universe size $u$ to satisfy
  $n^{6}\le u \le \poly(n)$.
  We adopt these restrictions for the following reasons:
  \begin{itemize}
    \item If $u$ is too small then bounding maxload is not
          interesting. For instance, if $u\in \Theta(n)$ linear
          hashing trivially achieves maxload $\bigO(1)$.
          It is standard to think of the universe as being much larger
          than the number of bins. Our specific choice $u\ge \Omega(n^{6})$ is
          arbitrary, but simplifies some analysis.
    \item If $u$ is too large then constant-time arithmetic
          operations becomes an unreasonable assumption. Thus, we require
          ${u\le \poly(n)}$.
  \end{itemize}
\end{rmk}

\section{Composite Moduli}
\label{sec:Z}
In this section we study the effect of replacing the standard
prime modulus in $\LH$ with a composite modulus.
In addition to being an intrinsically interesting question, we
see in \cref{sec:R} that $\LH$ with composite modulus arises
naturally in the study of $\LH$ over $\R$.
Primes are more well-behaved than composite numbers when used as
moduli because in $\F_p$ all non-zero elements have inverses.
However, \cref{fact:numdivs} and \cref{fact:toitent} indicate that
while $\Z_m$ could have a large quantity of elements with varying
degrees of ``degeneracy'', there are also guaranteed to
be a substantial number of relatively well-behaved elements.

To bound the extent to which $\Z_m$  is worse than $\F_p$ we
begin by defining \defn{Smart $\LH$} ($\SLH$) where the multiplier is chosen
randomly from $\Z_m^\times$ rather than $\Z_m$. Then, we show
that the maxload of modulus $m$ $\LH$ is at most
$\tau(m)$-times larger than that of modulus $m$ $\SLH$.
We finish the section by demonstrating that the standard $\bigO(\sqrt{n})$ maxload
bound for prime moduli $\LH$, and even Knudsen's beautiful
$\widetilde{\bigO}(n^{1/3})$ bound \cite{knudsen_linear_2017} can be translated with
several modifications to the composite integer setting.

Now we formally discuss our hash functions.
\begin{defin}
  Fix appropriate $n\in \N, m\in \Z$.
  We define two hash families parameterized by $a\in
    [m]_{\setminus 0}$ consisting of functions
  $h_{a} : [m] \to [n]$.
  \begin{enumerate}
    \item In \defn{Blocked Hashing}, denoted $\blockZm$,
          $h_{a}(x) = \floor{\frac{\posmod_m(ax)}{m / n}}.$
    \item In \defn{Strided Hashing}, denoted $\strideZm$,
          $h_{a}(x)= \posmod_n(\posmod_m(ax)).$
  \end{enumerate}

  We define $\blockFp, \strideFp$ to be  $\blockZm,\strideZm$ for
  $m=p\in \PRM$.
\end{defin}

In \cite{knudsen_linear_2017} Knudsen gives the necessary idea to
show an equivalence up to a factor-of-$2$ between $\blockFp$ and
$\strideFp$; this fact also follows immediately from our
\cref{prop:blockZisok}.
For composite integers the situation is
more delicate. In particular, if $\gcd(m,n)$ is large then
$\strideZm$ behaves extremely poorly for some
$X$ while $\blockZm$ does not.

\begin{prop}\label{prop:blockZsucks}
  Let $m=k\cdot n$ for some $k>n$.
  There exists an $n$-element set
  $X\subset [m]$ on which $\strideZm$ has maxload  $n$.
\end{prop}
\begin{proof}
  Let  $X=n\cdot [n]$. Then  $\posmod_n(\modm(ax)) = 0$
  for all $x\in X$ regardless of $a$. Thus, all $x\in X$
  always hash to bin $0$ so the maxload is $n$ deterministically.
\end{proof}

On the other hand, as long as $\gcd(m,n)$ is small then
$\strideZm,\blockZm$ achieve similar maxload.
\begin{prop}\label{prop:blockZisok}
  Let $m\perp n$. For any $n$ element set $X\subset [m]$ the expected maxload of
  $\strideZm$ and $\blockZm$ on $X$ differ by at most a factor-of-$2$.
\end{prop}
\begin{proof}
  Because $m,n$ are coprime $n$ has a multiplicative inverse
  $n^{-1}\in \Z_m$.
Assume $Y\subset X$ is the set of elements mapping to a fullest bin under
$\blockZm$ with $a=a_0$. We claim that for $a=\modm(n\cdot a_0)$, $\strideZm$
has maxload at least $|Y|/2$.
Indeed, let $k$ be the bin that $Y$ maps to under $a_0$.
We have 
\[ \ceil{k\frac{m}{n}} \le \modm(Ya_0) < (k+1)\frac{m}{n}. \] 
Consider the set 
 \[  Y' = n\cdot\left(\ceil{k\frac{m}{n}}+[\ceil{m/n}]\right). \] 
The difference between the maximum and the minimum elements of $Y'$ is at most
$m$. Thus, $\posmod_n(\modm(Y'))$ takes at most two distinct values.
In particular this implies that $\posmod_n(\modm(Ya_0n))$ takes on at most two
values, so there is a subset of $Y$ with size at least $|Y|/2$ that all hash to
the same bin under  $a=\modm(n a_0)$.

  Assume $S\subset X$ is the set of elements mapping to a fullest bin under
$\strideZm$ with $a=a_0$. Then, by similar reasoning to the above case, for
$a=\modm(n^{-1}\cdot a_0)$ $\blockZm$ has maxload at least  $|S|/2$.

  Multiplication by $n$ or $n^{-1}$ modulo $m$ permutes $\Z_m$. The result follows.
\end{proof}

\cref{prop:blockZisok} and \cref{prop:blockZsucks} teach us that
for composite integer $\LH$ it is more robust to consider
$\blockZm$ than $\strideZm$, but essentially equivalent as long
as $\gcd(m,n)$ is small. For the remainder of the paper we
restrict our attention to $\blockZm$ and $\blockFp$ which we
abbreviate to $\ZH_m$, $\FH_p$.
Now we formally define the variant of $\ZH_m$ that partially
solves the problem of $m$ being composite.
\begin{defin}\label{defn:slh}
  In \defn{Smart $\LH$} ($\SLH_m$) we randomly select
  $a\in \Z_m^{\times }$ and place $x\in [m]$ in bin $\floor{\frac{\modm(ax)}{m/n}}.$
\end{defin}
Surprisingly, we will show that the performance of $\ZH_m$ is not
too far from that of $\SLH_m$, especially if $m$ has few divisors.
We use the following notation:

\begin{defin}
  Let random variable $M_{\ZH}(m,X)$ denote the maxload incurred by
  $\ZH_m$ on $X$, and let $M_{\ZH}(m,n)$ denote the worst-case expected value of
  $M_{\ZH}(m,X)$ over all $n$-element sets $X\subset [m]$.
  Analogously define $M_{\SLH}(m,X), M_{\SLH}(m,n)$.
\end{defin}

\begin{theorem}\label{thm:LHSLH}
  Fix $m\ge n^{6}$ with $m\in \poly(n)$.
  Let $f$ be a monotonically increasing concave function with
  $M_{\SLH}(m_0,n_0)\le f(n_0)$ for all $n_0$ and all $m_0\ge n^{2.5}$.
  Then
  \[
  M_{\ZH}(m,n) \le \tau(m) \cdot f(n)+1.
\]
\end{theorem}
\begin{proof}
  Fix any $n$-element set $X\subset [m]$.
  For $d\mid m$, define $\setof{I_{i,d}}{i\in [d]}$ as the
  following partition of $[m]$ into $d$ size $m/d$ blocks:
  \[I_{i,d} = i\cdot m/d + [m/d].\]
  Define $X_{i,d} = X\cap I_{i,d}$
  and let $G_d$ be the event  $\gcd(a,m)=d$.
  We will bound the expected maxload of $\ZH_m$ by conditioning
  on $G_d$. However, if $\gcd(a,m)$ is very large then  $\ZH_m$ will
  necessarily incur large maxload; thus, we first exclude
  this case by showing it is very unlikely.
  For any $d\mid m$
  \begin{equation}\label{eq:prGd}
    \Pr[G_d]\le 1/d
  \end{equation}
  because there are $m/d$ multiples of $d$ in $[m]$.
  There are at most $n^{2.5}$ divisors $d\mid m$ with  $d\ge
    m/n^{2.5}$, because such divisors are in bijection with
  divisors $d'\mid m$ satisfying  $d'\le n^{2.5}$.
  By \cref{eq:prGd} each of these large divisors
  $d$ has $\Pr[G_d] \le n^{2.5}/m$. Thus we have:
  \begin{equation}\label{eq:gcdbignoway}
    \Pr[\gcd(a,m) \ge m/n^{2.5}] \le n^{2.5}\cdot \frac{n^{2.5}}{m} \le
    \frac{1}{n},
  \end{equation}
  where the final inequality follows by the assumption that $m\ge n^{6}$.
  By \cref{eq:gcdbignoway} the case $\gcd(m,a)\ge m/n^{2.5}$ contributes at most $1$ to
  the expected maxload.
  Let
  \[D_m = \setof{d\mid m}{d<m/n^{2.5}}.\]
  We claim the following chain of inequalities:

\begin{minipage}{0.4\textwidth}
\begin{align}
     &\E[M_{\ZH}(m,X) \mid \gcd(m,a)\in D_m] \nonumber \\
     &= \sum_{d\in D_m}\Pr[G_d]\cdot \E[M_{\ZH}(m,X) \mid G_d] \label{eqchainSLH1} \\
     &\le \sum_{d\in D_m} \frac{1}{d}\cdot \E[M_{\ZH}(m,X) \mid G_d] \label{eqchainSLH2} \\
     &\le \sum_{d\in D_m}\frac{1}{d}\sum_{i\in [d]}\E[M_{\ZH}(m,X_{i,d}) \mid G_d]\label{eqchainSLH3} \\
     &\le \sum_{d\in D_m}\frac{1}{d}\sum_{i\in [d]}\E[M_{\SLH}(m/d,X_{i,d})] \label{eqchainSLH4} \\
     &\le \sum_{d\in D_m}\frac{1}{d}\sum_{i\in
    [d]}M_{\SLH}(m/d,|X_{i,d}|) \label{eqchainSLH5} \\
     &\le \sum_{d\in D_m} \frac{1}{d}\sum_{i\in
    [d]}f(|X_{i,d}|) \label{eqchainSLH6} \\
     &\le \sum_{d\in D_m}f(n/d) \label{eqchainSLH7} \\
     &\le \tau(m)\cdot f(n) \label{eqchainSLH8}.
\end{align}
\end{minipage}%
\begin{minipage}{0.6\textwidth}
\begin{itemize}
    \item \cref{eqchainSLH1}: Law of total expectation.
    \item \cref{eqchainSLH2}: $\Pr[G_d]\le 1/d$, because there are $m/d$ multiples of $d$ in $[m]$.
    \item \cref{eqchainSLH3}: We can ``union bound''  because $\bigsqcup_{i\in [d]}X_{i,d} = X.$
    \item \cref{eqchainSLH4}: Recall that $X_{i,d}\subset I_{i,d}$, where
          $I_{i,d}$ is a contiguous interval of size $m/d$.
          Because we are conditioning on $\gcd(m,a)=d$,
          $\modm(a\cdot I_{i,d})$ consists of every $d$-th element of
          $[m]$ starting from $0$, i.e., is $\modm(d[m])$. Having
          elements which are spaced out by $d$ grouped into
          intervals of length $m/n$ per bin is
          equivalent to having elements spaced out by $1$
          grouped into intervals of length $(m/d)/n$ per
          bin. Formally this is because
          $\modm(x\cdot d\cdot j) = d\cdot \posmod_{m/d}(x\cdot j).$
          The restriction  $\gcd(m,a)=d$ can also be expressed as
          $\gcd(m/d, a/d)=1$, i.e., $a/d \perp m/d$.
          Hence, the expected maxload of $\ZH_m$ on $X_{i,d}$ conditional on
          $G_d$ is the same as the expected maxload of $\SLH_{m/d}$
          on $X_{i,d}$.
    \item \cref{eqchainSLH5}: $M_\SLH(m,n)$ is by definition the
          worst-case value of $\E[M_\SLH(m,X)]$ over all $n$-element
          sets $X$.
    \item \cref{eqchainSLH6}: By assumption $f$ is an upper
          bound on $M_{\SLH}$ as long as the modulus $m/d$ is
          sufficiently large. Because $d\in D_m$ we have  $m/d >
            n^{2.5}$, so the upper bound $f$ holds.
    \item \cref{eqchainSLH7}: $f$ is concave.
    \item \cref{eqchainSLH8}: $f$ is increasing, $\tau$ counts
          the divisors of $m$.
\end{itemize}
\end{minipage}

  We have shown the bound \cref{eqchainSLH8} for arbitrary
  $X$, so in particular the bound must hold for worst-case  $X$.
  Adding $1$ the for the event $\gcd(a,m)\ge m/n^{2.5}$ we have
  \[ M_{\ZH}(m,n) \le \tau(m)\cdot f(n)+1.\]
\end{proof}
\begin{rmk}
  \cref{thm:LHSLH} says that increasing concave bounds for
  $M_\SLH$ can be translated to bounds for $M_{\ZH}$ except
  weakened by a factor-of-$\tau(m)$.
  If $m$ is a power of $2$, a natural setting, then
  $\tau(m) = \log m$.
  Even for worst-case $m$ \cref{fact:numdivs} asserts $\tau(m)\le
    m^{o(1)}$. So, $\SLH_m$ and $\ZH_m$ have quite similar
  behavior.
\end{rmk}

Now we analyze the performance of $\SLH_m$. First we give an
argument based on the trivial $\bigO(\sqrt{n})$ bound for $\FH$.


\begin{theorem} \label{prop:sqrtnZ}
  $M_{\SLH}(m,n) \le \bigO(\sqrt{n\log\log n}).$
\end{theorem}
\begin{proof}
  We say $x,y$ \defn{collide}, with respect to $a=a_0$,
  if they hash to the same bin for $a=a_0$.
  As in the $\bigO(\sqrt{n})$ bound for $\FH_p$ we bound the
  maxload by counting the expected number of collisions and
  comparing this with the number of collisions guaranteed by a
  certain maxload.
  The difficulty in the proof for $\SLH_m$ is that the
  probability of  $x,y\in X$ colliding is not as simple as in
  $\FH_p$ where all pairs collide with probability $\bigO(1/n)$.
  To handle this we consider two types of pairs $x,y$:
  \begin{defin}
    Distinct $x,y\in X$ are
    \defn{linked} if
    $\gcd(x-y,m) > \ceil{m/n},$
    and \defn{unlinked} otherwise.
  \end{defin}

  \begin{claim}\label{clm:linkedcollidesqrt}
    Linked $x,y$ never collide.
  \end{claim}
  \begin{proof}
    If $\gcd(x-y,m)=k > \ceil{m/n}$ and $a$ is the randomly chosen
    multiplier from $\Z_m^{\times}$ then
    \[
    \gcd(\modm(a\cdot (x-y)),m) = k
  \]
    because $a\perp m$.
    But then
    $\circabs_m(ax - ay) \ge k > \ceil{m/n},$
    so $x,y$ fall in different bins.
  \end{proof}
  \begin{claim}\label{clm:unlinkedcollidesqrt}
    Unlinked $x,y$ collide with probability at most
    $\bigO\left(\frac{\log\log n}{n}\right).$
  \end{claim}
  \begin{proof}
    When $x,y$ are unlinked, $\modm(a(x-y))$ would have probability $\bigO(1/n)$
    of landing in each bin if $a$ were chosen uniformly from $\Z_m$.
    However, in $\SLH_m$ this uniformity is not obvious.
    Fortunately, \cref{fact:toitent} ensures that each $a\in
      \Z_m^{\times}$ occurs with probability at most $\frac{2\log\log
        m}{m}$, which is not much larger than $\frac{1}{m}$.
    In particular, this implies that the probability of $x,y$
    colliding is at most
    \[
    \bigO(1/n)\cdot m \cdot \frac{2\log\log m}{ m} \le
      \bigO\left(\frac{\log\log n}{n}\right).
    \]
  \end{proof}

  Now that we have shown
  \cref{clm:linkedcollidesqrt} and \cref{clm:unlinkedcollidesqrt}
  the proof continues in the same
  way as for $\FH_p$.
  If maxload is $m$ then there must be at least
  $\binom{m}{2} = \Theta(m^{2})$ collisions.
  By Jensen's inequality and the convexity of $x\mapsto x^2$ we
  have $\E[M_{\SLH}(m,X)^2]\ge \E[M_{\SLH}(m,X)]^2.$
  We can also count the expected number of collisions directly
  using \cref{clm:linkedcollidesqrt} and \cref{clm:unlinkedcollidesqrt};
  doing so, we find that the expected number of collisions is
  $\bigO(n \log\log n).$
  Comparing our two methods of counting collisions gives:
  \[
  \E[M_{\SLH}(m,X)]\le \bigO(\sqrt{n\log\log n})
\]
  for any $X$, and in particular for worst-case $X$.
\end{proof}

In \cref{sec:formalpfZm13} we strengthen \cref{prop:sqrtnZ} to 
\begin{theorem}\label{thm:Zm1_3}
  $M_{\SLH}(m,n)\le \widetilde{\bigO}(n^{1/3}).$
\end{theorem}
The proof is a modification of Knudsen's proof
\cite{knudsen_linear_2017} of the corresponding bound for
$\FH_p$ with modifications similar to those used in \cref{prop:sqrtnZ}.





\begin{cor}\label{cor:translate}
  $M_{\ZH}(m,n) \le n^{1/3 + o(1)}.$
\end{cor}
\begin{proof}
  This follows immediately from using \cref{thm:Zm1_3} in
  \cref{thm:LHSLH}, which is valid because $n^{1/3}$ is a concave
  and increasing function of $n$.
  \footnote{In general one needs to slightly modify the universe size
    in order to apply \cref{thm:LHSLH}. However, \cref{thm:Zm1_3}
    as proved in \cref{sec:formalpfZm13} only requires the universe
    size $m> n^{2.5}$. Thus, such a modification is not necessary here.}
\end{proof}

In \cref{thm:Zm1_3} we have translated the state-of-the-art maxload
bound for $\FH_p$ to $\SLH_m$ by altering Knudsen's proof.
This is tentative evidence that composite modulus $\LH$ may
achieve similar maxload to prime modulus $\LH$ in general.
We leave proving or refuting this as an open problem:
\begin{question}\label{question:equivalenceFZ}
  Are the worst-case maxloads of $\FH_p$ and $\SLH_m$ the same
  up to a factor-of-$n^{o(1)}$?
\end{question}

\section{Hashing with Reals}
\label{sec:R}

In this section we consider a continuous variant of $\LH$. Formally:
\begin{defin}
  Fix universe size $u\in \N$. As always, we require
  $n^{6} \le u \le \poly(n)$ (\cref{rmk:assumesize}). In
  \defn{Blocked Real Hashing} ($\RH_u$) we randomly select real $a\in
    (0,1)$ and place $x\in [u]$ in bin
  $\floor{\frac{\posmod_1(ax)}{ 1  / n }}.$
\end{defin}

One interesting similarity between $\RH_u$ and  $\FH_p$ is
that all $a\in (0,1)$ are invertible modulo $1$ over $\R$ just
as all $a\in \pnozero$ are invertible in $\F_p$.
In this section we show that $\RH_u$ actually behaves like a
version of $\SLH_m$ with a randomized modulus.
In particular, our lower bound on $\RH_u$'s performance is
relative to the following variant of $\SLH_m$:
\begin{defin}
  Fix $m\in \N$. In \defn{Random Linear Hashing} ($\randH_m$) we
  randomly select $k\in [m/2,m]\cap \Z$ and hash with $\SLH_k$'s
  hash function (\cref{defn:slh}).
\end{defin}
\begin{prop}\label{prop:randHstrictbetter}
  $\randH_m$ has expected maxload at most $\max_{k\in [m/2,m]\cap
      \Z}  2M_{\SLH}(m,k).$
\end{prop}
\begin{proof}
  Fix $X$, condition on some $k$. Partition $X$ into
  $X_1 = X\cap [k],\;\; X_2=X \cap [k,2k].$
  The maxload on $X_1,\posmod_k(X_2)\subset [k]$ individually is at most $M_{\SLH}(m,k)$.
  Adding the maxload on $X_1$ and $X_2$ gives an upper bound on the
  total maxload.
\end{proof}

Now we connect $\RH_u$ and $\ZH_m, \randH_m$.
\begin{theorem}
  \label{thm:itisreal}
  Let $M_{\randH}$ denote the expected maxload of
  $\randH_{\lfloor\sqrt{nu}\rfloor}$ on a worst-case $X\subset
    [\lfloor\sqrt{nu}\rfloor]$, let $M_{\max}$ denote the maximum over
  $k\in (\sqrt{u}, nu]\cap \Z$ of the expected maxload
  of $\SLH_k$ for worst-case $X\subset [k]$. Let $M_\R$ denote
  the expected maxload of $\RH_u$ for worst-case $X\subset
    [u].$ Then,
    \[
   \Omega\paren{\frac{M_{\randH}}{\log\log n}} \le M_\R \le \bigO(M_{\max}).
\]
\end{theorem}


The remainder of the section  is devoted to proving \cref{thm:itisreal}.
The link between integer hashing and real hashing begins to
emerge in the following lemma:
\begin{lemma}\label{lem:restrictQ}
  Fix $k\in \N$. Assume $a=c/k$ for random $c\in \Z_k^{\times}$ and take
  $X\subset [u]$. $\RH_u$ conditional on such $a$ achieves
  the same maxload on $X$ as hashing $X$ with $\SLH_k$'s hash
  function with multiplier $c$.
\end{lemma}
\begin{proof}
  For any $x\in [u]$ the value $\posmod_1(xc/k)$ is an integer multiple of
  $1/k$. In particular, $\RH_u$ places $x$ in bin
  \begin{equation}\label{eq:thatbin}
    \floor{\frac{\posmod_1(x c/k)}{1/n}} = \floor{\frac{\posmod_k(x c)}{k/n}}.
  \end{equation}
  $\SLH_k$'s hash function places $x$ in the same bin
  as \cref{eq:thatbin}.
\end{proof}

In isolation \cref{lem:restrictQ} is not particularly useful
because $a\in \Q$ occurs with
probability $0$. However, in \cref{clm:approxepx} we show that if
$a$ is very close to a rational number then we get approximately
the same behavior as in \cref{lem:restrictQ}.

\begin{lemma}\label{clm:approxepx}
  Fix $k\in [nu]_{\setminus{0}}$.
  Let $a = c/k+\eps$ for integer $c\perp k$ and real $\eps\in [0,
    \frac{1}{nu})$. Let $a'=c/k$.
  The maxload achieved by $\RH_u$ using $a$ and
  $a'$ differ by at most a factor-of-$2$.
\end{lemma}
\begin{proof}
  We refer to the interval $[i/n, (i+1)/n)$ as \defn{pre-bin} $i$; if element
  $x$ has $\posmod_1(ax)$ lie in pre-bin $i$, then
  $x$ is mapped to bin $i$ by $a$.

  For any $x\in [u]$ we have that $\eps x$, the difference
  between $\posmod_1(ax), \posmod_1(a'x)$ satisfies
  $\eps\cdot x< 1/n.$
  Thus, $\posmod_1(ax), \posmod_1(a'x)$ either lie in the same or adjacent pre-bins.
  Say that bin $j$, a fullest bin for $a$, has $M$ elements.
  Then, either bin $j$ or bin $\posmod_n(j+1)$ has at least
  $M/2$ elements for $a'$.
  If bin $j'$, a fullest bin for $a'$, has $M'$
  elements then either bin $j'$ or bin $\posmod_n(j-1)$ has at
  least $M'/2$ elements for $a$. Thus, the maxload with $a,a'$
  differ by at most a factor-of-$2$.
\end{proof}
All ${a\in (0,1)}$ will be within $\frac{1}{nu}$ of some $a'\in
  \Q$, and in particular some fraction with denominator at most $nu$.
This motivates the following definition:

\begin{defin}\label{defn:fclaims}
  For $k\in [nu]_{\setminus 0}$ we define ${I(k) \subset (0,1)}$ to be the
  set of ${a\in (0,1)}$ which are at most $\frac{1}{nu}$ larger
  than some reduced fraction with denominator $k$. That is,
  \[
  I(k) = \bigcup_{c\in \Z_k^{\times}} \paren{c/k +
\left[0,\frac{1}{nu}\right]} \cap (0,1).\]
      We say that $k$
  \defn{claims} the elements of $I(k)$.
  For $a\in (0,1)$ let $F(a)$ denote be the smallest $k\in \N$ so
  that $k$ claims $a$.
  For $a\in I(k)$ we say that  $k$ \defn{obtains} $a$ if $F(a)=k$
  and we say that  $a$ is \defn{stolen} from $k$ if $F(a)<k$.
\end{defin}

Combining \cref{lem:restrictQ} and \cref{clm:approxepx}, $F(a_0)=k$
intuitively means that if $a=a_0$ then $\RH$
will behave like integer hashing with modulus $k$. Thus, to
bound $M_\R$ we seek to understand $F$.
\begin{lemma}\label{clm:prdistrunderstand}
  For $k\le \lfloor\sqrt{nu}\rfloor$,
  $\Pr[F(a)=k] = \frac{\phi(k)}{nu}.$
\end{lemma}
\begin{proof}
  Immediately from \cref{defn:fclaims}
  \begin{equation}\label{eq:fakprub}
    \Pr[F(a)=k]\le |I(k)| = \frac{\phi(k)}{nu}.
  \end{equation}
  For distinct $k_1,k_2\le \lfloor\sqrt{nu}\rfloor$ and appropriate numerators
  $c_1\in [k_1],c_2\in [k_2]$
  we have
  \begin{equation}
    \label{eq:disjointdudes}
    \abs{\frac{c_1}{k_1}-\frac{c_2}{k_2}}>\frac{1}{nu}
  \end{equation}
  because $k_1k_2 < nu$ while $c_1k_2-c_2k_1 \in
    \Z\setminus\set{0}$, and thus is at least $1$ in absolute value.
  \cref{eq:disjointdudes} means that for any $k\le
    \floor{\sqrt{nu}}, a\in I(k)$, $a$ will not be stolen from $k$
  because all reduced fractions of denominator $k'<k$ are
  sufficiently far away from all reduced fractions of denominator
  $k$. In other words,
    $I(k_1) \cap I(k_2) = \varnothing.$
    This implies that \cref{eq:fakprub} is tight for $k\le
    \floor{\sqrt{nu}}$, which gives the desired bound on $\Pr[F(a)=k]$.

\end{proof}

The understanding of $F$ given by \cref{clm:prdistrunderstand}
is sufficient to establish our lower bound on $M_\R$.
\begin{cor}
  \label{cor:lb}
  $M_\R \ge \frac{M_{\randH}}{20\log\log n}.$
\end{cor}
\begin{proof}
  Fix integer ${k\in [\lfloor\sqrt{nu}\rfloor/2, \lfloor\sqrt{nu}\rfloor]}$.
  Using \cref{fact:toitent} on \cref{clm:prdistrunderstand} gives
  \begin{equation}
    \Pr[F(a)=k]  \ge \frac{1}{nu}\frac{\floor{\sqrt{nu}}/2}{2
    \log\log (\floor{\sqrt{nu}}/2)}                                        
                 \ge \frac{1}{\floor{\sqrt{nu}}} \cdot \frac{1}{5\log\log
      n}\label{eq9999},
  \end{equation}
  where the simplification in \cref{eq9999} is due to the
  asymptotic nature of our analysis (\cref{rmk:assumesize}).
  Fix ${X\subset [\floor{\sqrt{nu}}]\subset [u]}$. We make two observations:
  \begin{itemize}
    \item Each modulus $k \in [\floor{\sqrt{nu}}/2,
              \floor{\sqrt{nu}}]\cap \Z$ is selected by
          $\randH_{\dfloor{\sqrt{nu}}}$ with probability
          $2/\dfloor{\sqrt{nu}}$ which is at most $10 \log\log n$
          times larger than $\Pr[F(a)=k]$. \item Conditional on
          $F(a)=k$ $\RH_u$ achieves expected maxload at least
          $1/2$ of the expected maxload of $\randH_{\dfloor{\sqrt{nu}}}$
          conditional on $\randH$ having modulus $k$; this follows by
          combining \cref{clm:approxepx} and \cref{lem:restrictQ}.
  \end{itemize}
  Combining these observations gives the desired bound on
  $M_\R/M_\randH$.
\end{proof}

Now we aim to show an upper bound on $M_\R$.
\cref{clm:approxepx} combined with \cref{lem:restrictQ} shows
that $\RH_u$ is essentially equivalent to using the $\SLH$ hash function
but with a modulus chosen according to some probability
distribution; we call $F(a)$ the \defn{effective integer
  modulus}.
However, the distribution of the effective integer modulus is
very different from the distribution of moduli for $\randH$. One
major difference is that in $\randH_m$ the randomly selected modulus is always
within a factor-of-$2$ of the universe size $m$. However, in
$\RH_u$ the effective integer modulus is likely of size
$\Theta(\sqrt{nu})$ which is much smaller than the universe
size $u$. A priori this might be concerning: could some choice
of $X$ result in many items hashing to the same bin by
virtue of being the same modulo the effective integer modulus?
\cref{lem:nothingcollides} asserts that this is quite unlikely.
Before proving \cref{lem:nothingcollides} we need to obtain more
bounds on $F$. We do so by use of \defn{Farey
  Sequences} (see \cite{hardy1979introduction} for an excellent exposition).

\begin{defin}
  For $k\in \N$, a \defn{$k$-fraction} is some rational $c/k \in
    [0,1]$ with $c\in \Z, c\perp k$.
  For $m\in \N$, the \defn{$m$-Farey sequence} $\mathfrak{F}_m$
  is the set of all $k$-fractions for
  all $k\le m$ listed in ascending order.
  For example $\mathfrak{F}_5$ is the sequence
  \[
  \frac{0}{1}, \frac{1}{5}, \frac{1}{4}, \frac{1}{3}, \frac{2}{5}, \frac{1}{2}, \frac{3}{5}, \frac{2}{3}, \frac{3}{4}, \frac{4}{5},
\frac{1}{1}.\]
\end{defin}

A fundamental property of the Farey sequence concerns the
difference between successive terms of the sequence (see \cite{hardy1979introduction}).
\begin{fact}\label{fact:fareydist}
  Fix $m\in \N$. Let $c/k, c'/k'$ be adjacent fractions in
  $\mathfrak{F}_m$.
  Then
  $\abs{\frac{c}{k} - \frac{c'}{k'}} = \frac{1}{kk'}.$
\end{fact}

We further classify the neighbors of Farey fractions in the
following lemma:
\begin{lemma}\label{lem:fareyneighbors}
  Fix $m\in \N$. Let $S\subset \mathfrak{F}_m$ be the set of successors of
  $m$-fractions in $\mathfrak{F}_m$.
  For each $\ell\in [m], \ell\perp m$ there is precisely one
  $\ell$-fraction in
  $S$. For ${\ell\in [m]}, \ell\not\perp m$ there are no
  $\ell$-fractions in $S$.
\end{lemma}
\begin{proof}
  Fix $\ell\in [m], \ell\perp m$.
  Then the equation
  \begin{equation}\label{eq:liminv}
    \ell i\equiv -1\mod m
  \end{equation}
  has a solution $i\in \Z_m^{\times }$.
  Let $\lambda \in \N$ so that $i\ell = \lambda m - 1$. Then,
  \[
  \frac{i}{m} = \frac{\lambda m / \ell - 1/\ell}{m}=
\frac{\lambda}{\ell}-\frac{1}{m \ell}.\]
  By \cref{fact:fareydist} the predecessor of
  $\frac{\lambda}{\ell}$ in $\mathfrak{F}_m$ is at
  least $\frac{1}{m \ell}$ smaller than $\frac{\lambda}{\ell}$. Hence,
  there are no fractions in $\mathfrak{F}_m$ between
  $\lambda/\ell$
  and  $i/m$. That is, $\lambda/\ell$ is the successor of $i/m$.

  In fact, in \cref{eq:liminv} we took $i\equiv -\ell^{-1}\mod
    m$. Clearly
  $\setof{-\ell^{-1}}{\ell\in \Z_m^{\times }} = \Z_m^{\times }.$
  Thus, we have already identified the successor of every
  $m$-fraction.
  So for $\ell\not\perp k$ there are no $\ell$-fractions in $S$.
\end{proof}

Now analyze $\Pr[F(a) = k]$ using Farey sequences.
\begin{lemma}\label{lem:fareyopfak}For all $k\in [nu]$
  $\Pr[F(a)=k]\leq \bigO\paren{\frac{1}{\sqrt{nu}}}.$
\end{lemma}
\begin{proof}
  For $k\le \floor{\sqrt{nu}}$ the conclusion is immediate by
  \cref{clm:prdistrunderstand}. Fix integer $k\ge \sqrt{nu}$.
  To bound the measure of $\setof{a\in (0,1)}{F(a)=k}$ we can
  take the measure claimed by $k$ and subtract the measure
  stolen from $k$ by any $k'<k$.


  Fix a $k$-fraction $c/k$. The interval touching $c/k$ claimed
  by $k$ is $c/k+[0,1/(nu)]$.
  The amount of this interval which is stolen is determined by
  the distance from $c/k$ to $(c/k)$'s successor in $\mathfrak{F}_k$.
  Let $i<k$ denote the denominator of $(c/k)$'s successor.
  As depicted in \cref{fig:fareypf} there are two cases:
  \begin{itemize}
    \item If the successor is close to $c/k$ then all but a
          length-$\frac{1}{ik}$ prefix of the interval is stolen, where
          $\frac{1}{ik}$ is the distance to the successor by
          \cref{fact:fareydist}.
    \item If the successor of $c/k$ occurs after distance more than
          $\frac{1}{nu}$ then it will not steal any of the interval.
  \end{itemize}

  \begin{figure}[h]
    \centering
    \includegraphics[width=0.45\textwidth]{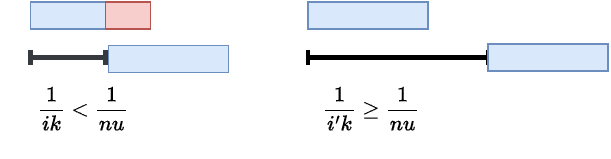}
    \caption{$k$-fraction and its successor.\\
      Left: close successor. Right: far successor.\\
      Blue box: obtained $a$.  Red box: stolen $a$.}
    \label{fig:fareypf}
  \end{figure}

  Now, we bound the measure obtained by $k$ by summing over the two
  cases represented in \cref{fig:fareypf}.
  \begin{equation}
    \Pr[F(a)=k]  \le\frac{nu}{k} \cdot \frac{1}{nu} + \sum_{i=
    \floor{nu/k}}^k \frac{1}{ik}                                
                 \le \paren{2+\ln \paren{\frac{k^2}{nu}}}
    \frac{1}{k}.\label{eq:lnudk2}
  \end{equation}
  Let $k=\alpha\sqrt{nu}$ for some $\alpha\ge 1$.
  Using $\alpha$ in \cref{eq:lnudk2} gives:
  \[\frac{2+2\ln\alpha}{\alpha} \cdot \frac{1}{\sqrt{nu}} \le
  \bigO\paren{\frac{1}{\sqrt{nu}}},\]
  the desired bound on $\Pr[F(a)=k]$.

\end{proof}

Now we are prepared for the following lemma:
\begin{lemma}\label{lem:nothingcollides}
  With probability at least $1-1/n$ all
  pairs $x,y\in X$ with $x\neq y$ satisfy
  $x\not\equiv y \mod F(a).$
\end{lemma}
\begin{proof}
  Take distinct $x,y\in X$. We say $x,y$ \defn{collide} if
  $x\equiv y \bmod F(a)$. If $x,y$ collide we must
  have ${F(a) \mid (x-y)}$. By \cref{fact:numdivs}, $x-y$ has at most
  $u^{o(1)}$ divisors. By \cref{lem:fareyopfak} $F(a)$ will be one
  of these divisors with probability at most
  $u^{o(1)}/\sqrt{nu}$. That is, $x,y$ collide with probability
  at most $u^{o(1)}/\sqrt{nu}$.
  By linearity of expectation the expected
  number of pairs $x,y$ which collide is at most
  \[
  \frac{\binom{n}{2}u^{o(1)}}{\sqrt{nu}} \le
    \frac{n^{2}n^{o(1)}}{n^{7/2}} \le \frac{1}{n^{3/2-o(1)}} \le
    \frac{1}{n},
  \]
  which we have simplified using \cref{rmk:assumesize}.
  The number of colliding pairs is a non-negative
  integer random variable. Thus, by Markov's inequality the
  probability of having at least $1$ collision is at most $1/n$.
  Equivalently, with probability at least $1-1/n$ there are $0$
  colliding pairs.
\end{proof}

\begin{cor}\label{cor:ub}
  $M_\R \le \bigO(M_{\max}).$
\end{cor}
\begin{proof} If the effective integer modulus $F(a)$ of $\RH$ is very
  small then $\RH$ may perform quite poorly. Luckily, $F(a)$ is
  likely not too small. By \cref{clm:prdistrunderstand} we have
  \[
  \Pr[F(a) \le \sqrt{u}] \le \sum_{k\le \sqrt{u}}\frac{k}{nu}
    \le 1/n,
  \]
  so the contribution to $M_\R$ of $a$ with $F(a)\le \sqrt{u}$ is
  $\bigO(1)$. Thus, it suffices to consider $F(a) \in (\sqrt{u},
    nu]$.
  Fix $X\subset [u]$, and condition on $F(a)=k$ for some $k\in (\sqrt{u},
    nu]$.
  By \cref{clm:approxepx} and \cref{lem:restrictQ} the expected
  maxload of $\RH_u$ on  $X$ is at most twice the expected
  maxload if we hash $X$ with $\SLH_k$'s hash function.
  By \cref{lem:nothingcollides} with probability at least $1-1/n$
  we have
  $|\posmod_k(X)| = n.$
  The case where $|\posmod_k(X)| < n$ contributes at most
  $\bigO(1)$ to the maxload. Otherwise we have a set
  $\posmod_k(X) \subset [k]$ on which  $\RH_u$ has expected
  maxload at most twice that of $\SLH_k$.
  This yields the desired bound on $M_\R$.
\end{proof}

Together \cref{cor:lb}, \cref{cor:ub} prove \cref{thm:itisreal}.
As a bonus, applying \cref{thm:Zm1_3} to \cref{thm:itisreal} gives:
\begin{cor}
  The expected maxload of $\RH_u$ is at most $\widetilde{\bigO}(n^{1/3}).$
\end{cor}

\section{Two Bins}
\label{sec:twobins}
In this section we consider the simplest possible setting for
$\LH$: hashing $n$ items to $2$ bins.
If $\LH$ performs well on $n$ bins then it is reasonable to conjecture that the
maxload of $\LH$ with two bins is tightly concentrated around $n/2$.  We propose
as an open problem proving a Chernoff-style concentration bound on the maxload:
\begin{conj}\label{conj:reductiontwobins}
  Two bin $\LH$ incurs maxload larger than $n/2+k\sqrt{n}$
  with probability at most $2^{-\Theta(k^2)}+1/p.$
\end{conj}
We propose that showing a result such as \cref{conj:reductiontwobins}
may be relatively tractable compared to analysis of the full $n$ bin case.
Establishing such a conjecture would constitute the strongest evidence to date
that $\LH$ is a good load-balancing function.
As partial progress towards understanding $\LH$ in the $2$-bin case we analyze
its expected maxload.
Formally our hash function is defined as follows:
\begin{defin}
  Let $p\in \PRM$. In \defn{Multiplicative Two Bin} $\LH$
  ($\MLH$) we choose random $a\in \pnozero$ and place $x\in
    [p]$ in bin
  $\floor{\frac{\modp(ax)}{p/2}} \in \set{0,1}.$
\end{defin}

We prove
\begin{theorem}\label{thm:dontneedb}
  $\MLH$ has expected maxload at most
  $n/2+\widetilde{\bigO}(\sqrt{n}).$
\end{theorem}

The proof of \cref{thm:dontneedb} uses the standard technique of analyzing the expected
number of \defn{collisions}: pairs of elements that hash to
the same bin. However, without pairwise independence computing
the expected number of collisions is challenging.
In fact, some elements collide with probability much
larger than $1/2$. For instance, $1$ and $3$ collide with probability
$2/3$. 
More generally for any small integer $k$, $1$ and $2k+1$ will
collide with probability approximately $(k+1)/(2k+1) > 1/2$.
If small odd numbers were the only numbers with probability
much larger than $1/2$ of colliding with $1$ then the analysis
would be fairly easy. However, there can be other numbers which
are very likely to collide with $1$. For instance, imagine
$x\in[p]$ satisfies  $3x\equiv 1 \mod p$. Then $1$ and $x$ also
have an approximately $2/3$ chance of colliding.

Our bound on the expected number of collisions intuitively
works as follows: for any particular element $x\in [p]$ there
are very few $y\in [p]$ where $x,y$ collide with probability
much larger than $1/2$. By symmetry (or more precisely the
existence of multiplicative inverses in $\F_p$), it does not
matter which $x$ we choose to compare with. So it will suffice
to analyze the probability of elements $y$ colliding with
$x=1$.

For sake of combinatorics we work with the following
transformed version of collision probabilities:
\begin{defin}
  The \defn{overlap} of $x\in [p]$ is the number of $a\in \halfp$
  such that  $1,x$ collide.
  Equivalently, the overlap of $x$ is the number of $a\in \halfp$ where
  $\modp(ax) < p/2.$
  The \defn{excess overlap} of $x$, denoted $\lap(x)$, is the
  overlap of $x$ minus $p/4$.
  The \defn{contribution} of a set ${A\subset \halfp}$ to $\lap(x)$
  is the difference between the number of $a\in A$
  with $\modp(ax) <p/2$ and the number of  $a\in A$ with
  $\modp(ax) > p/2$.
  We will bound $\lap(x)$ by partitioning $\halfp$ into
  disjoint subsets $A_1,A_2,\ldots$ and summing the
  contributions of each $A_i$ to $\lap(x)$.
\end{defin}

Now we give a bound on $\sum_{x\in X}\lap(x)$. Our key insight is
that $\lap(x)$ is best understood by finding a small
number $m$ so that $k=\circabs_p(xm)$ is small and using this
$m,k$ to partition $\halfp$ into parts that each have small
contribution to $\lap(x)$. 

\begin{lemma}\label{lem:pigeons}
  For any $x\in \Z_p$ there exist $m\in [n], k\in [\ceil{p/n}],
  \sigma\in \pm 1$ such that $x = \modp(\sigma m^{-1}k)$.
\end{lemma}
\begin{proof}
  By the pigeonhole principle the set $\setof{\modp(x\cdot i)}{i\in
  [n]}$ must have two numbers within distance $p/n$ of each
  other.
  Let $i_1,i_2\in [n]$ be distinct indices such that
  $\modp(xi_1-xi_2)\in [0, p/n]$.
  Set $m = |i_1-i_2| \in [n]$, and set $\sigma$ to be the sign of
  $i_1-i_2$. Then $\modp(x\sigma m)\in [\ceil{p/n}]$. Define $k$ to be
  $\modp(x\sigma m)$. Clearly we have $x = \modp(\sigma m^{-1}
  k)$, for $m,k,\sigma$ with the desired properties.
\end{proof}

\begin{lemma}\label{lem:epicbound}
  Let $x= \modp(\sigma m^{-1}k)$ for
  $\sigma\in \pm 1,m\in [n],k \in [\ceil{p/n}]$.
  Then, \[\lap(x) \le \bigO\left(k+\left(m + \frac{p}{mk}\right)\cdot
    \gcd(k,m)\right).\]
\end{lemma}
\begin{proof}
  We partition $\halfp$ into $m$ \defn{groups}, where for each
  $i\in [m]$, group $i$ consists of the values $G_i = {(m\Z+i)\cap \halfp}$.
  We further split groups into \defn{cycles}, where cycle $C_{i,j}$
  is defined to be 
  \[
    C_{i,j} = \setof{m j' + i}{j' \in [j\ceil{p/k}, (j+1)\ceil{p/k})} \cap
    \halfp \subseteq G_i.
  \]
  We say that a cycle $C_{i,j}$ is a \defn{full cycle} if
  $|C_{i,j}| = \ceil{p/k}$.
  For any full cycle $C_{i,j}$ the set $\modp(x\cdot C_{i,j})$
  consists of $\ceil{p/k}$ points, with consecutive points
  separated by distance $k$; this is due to the fact that
  $x=\modp(\sigma m^{-1} k)$. In particular, this means that the
  points go slightly past a full revolution of the circle $\Z_p$.
  Thus, each full cycle contributes at most $\bigO(1)$ to $\lap(x)$.
  The total contribution to $\lap(x)$ from all full cycles is thus 
  bounded by 
  \[
    \sum_{i\in [m]} \floor{\frac{|G_i|}{\ceil{p/k}} } \cdot
    \bigO(1)\le \bigO\left(m \cdot \frac{p/m}{p/k}\right) \le \bigO(k).
  \]
  Now it suffices to bound the contribution from non-full
  cycles. Observe that each group $G_i$ has a single non-full
  cycle, which we call its \defn{final cycle}, or $F_i$ (if $G_i$ has no non-full
  cycle, then its final cycle is $\varnothing$). When bounding the
  contribution from the final cycles it is no longer a good idea
  to analyze the groups separately. 
  Instead, for each $j\in [|F_{m-1}|]$ we will group together the
  $j$-th largest values in each of the $F_i$'s into a
  \defn{step}, denoted $S_j$.
  Note that this might not quite capture all the points in all
  the final cycles because there may be some $i$ such that
  $|F_{i}| = |F_{m-1}|+1$. 
  To handle this we remove the largest value from each final
  cycle $i$ with $|F_i| = |F_{m-1}|+1$. This results in contribution
  $\bigO(m)$ to $\lap(x)$.
  This done, it suffices to consider the contribution of
  $\bigcup_{i\in [L]}S_i$ where $L$ is the number of steps.
  Note that $L$ satisfies $L\le \ceil{p/k}$, or else the final
  cycles have enough points that they would have been full cycles.
  Now we prove a powerful structural result about the steps.
  \begin{claim}\label{clm:structureSi}
    There exists $\Delta \in \Z, \lambda\in \Z_m^*$ and $\delta_j\in
    [-k,k]$ for $j\in [m]$ such that 
    \begin{equation}\label{eq:S0goal}
     \modp(x\cdot S_0)=\setof{\modp\left(\Delta + p\frac{k\lambda j}{m} + \delta_j\right)}{j\in [m]}.
    \end{equation}
  \end{claim}
  \begin{proof}
    $\modp(m^{-1}m) = 1$. 
    So, as an integer $m^{-1}$ can be written in the form
    $\frac{1+\lambda p}{m}$ for some integer $\lambda$.
    We claim that $\lambda\perp m$. If not, then we would have
     \[
    \frac{m}{\gcd(\lambda,m)} \cdot \frac{1+\lambda p}{m} =
    \frac{1}{\gcd(\lambda, m)} + p\frac{\lambda}{\gcd(m,\lambda)}
    \notin \Z,
    \]
    which is clearly impossible.
    Thus, there is $\lambda\perp m$ and some $\alpha\in \N$ so that 
    \[
      \modp(x\cdot S_0) = \setof{\modp\left(\sigma m^{-1} k\cdot (\alpha m +
      \beta)\right)}{\beta\in [m]} = \setof{\modp\left(\sigma
      k \alpha + p\frac{k\lambda \sigma\beta}{m} +
  \frac{\sigma k\beta}{m}\right)}{\beta\in [m]}.
    \]
    The term $\sigma k \alpha$ is the $\Delta$ from \cref{eq:S0goal}.
    The term $p\frac{k\lambda \sigma\beta}{m}$ is the
    $p\frac{k\lambda j}{m}$ from \cref{eq:S0goal} (where we may
    eliminate $\sigma$ by re-indexing if $\sigma=-1$).
    Finally, the term $\frac{\sigma k\beta}{m}$ is the $\delta_j$
    from \cref{eq:S0goal}, and it does indeed satisfy $\delta_j
    \in [-k,k]$, as required.
    
  \end{proof}
  We call the points $\mathbb{A} = \modp(p\frac{k}{m}\Z)$ \defn{anchors}.
  Observe that $|\mathbb{A}| = m/\gcd(m,k)$.
  A \defn{rotation} of the anchors is the set
  $\modp(\mathbb{A}+\Delta)$ for some $\Delta\in \Z$.
  In \cref{clm:structureSi} we showed that there is some
  $\Delta\in \Z$ such that for each $y\in
  \modp(\mathbb{A}+\Delta)$ $\modp(x S_0)$ has $\gcd(m,k)$ points which
  are very close to $y$.
  The other important fact that will let us control the steps is
  that $\modp(x S_{i+1}) = \modp(xS_i + \sigma k).$ 

  Now we consider two cases. 
  The easier case is if $m$ is even.
  In this case, we have the helpful property that for any rotation of
  the anchors there are an equal number of anchors in
  $[0,p/2)$ and in $[p/2, p)$.
  If the points in $\modp(x\cdot S_i)$ where \emph{actually}
  located at the anchors then the
  contribution from $\bigcup_{i\in [L]}S_i$ would be \emph{zero}.
  However, the points are allowed to deviate by a small amount
  from the anchors. Evidently this only results in a
  contribution to $\lap(x)$ at most $\bigO(1)$ times per every
  $\ceil{p/(km)}$ consecutive steps. Thus, the contribution from
  $\bigcup_{i\in [L]}S_i$ is bounded by:
  \[
  \bigO(\gcd(k,m))\cdot \frac{L}{p/(km)} \le \bigO(m\gcd(k,m)).
  \]

  Now we consider the case that $m$ is odd.
  First, we claim that the contribution of any single step $S_i$ is at
  most $\bigO(\gcd(m,k))$.
  This is by \cref{clm:structureSi}: the points in $S_i$ are
  concentrated around the anchors, the number of anchors in
  $[0,p/2)$ and $[p/2, p)$ differ by at most $1$, and
  hence the contribution of $S_i$ is at most $\bigO(\gcd(m,k))$.
  Now we show that certain large groups of steps called
  \defn{revolutions} have contribution $\bigO(\gcd(m,k))$.
  Revolution $i$, or $R_i$, is the steps $S_{j}$ for each $j\in [L] \cap
  [\ceil{p/(km)}i, \ceil{p/(km)}(i+1)).$
  We say $R_i$ is a \defn{full revolution} if $R_i$ consists of
  $\ceil{p/(km)}$ steps.
  There is at most one (non-empty) non-full revolution.
  We bound the contribution from the non-full revolution by
  $\bigO(\ceil{p/(km)} \gcd(k,m))$, using our earlier observation
  that each individual step has contribution as most
  $\bigO(\gcd(m,k))$.
  Now we argue that if $R_i$ is a full revolution then the
  contribution of $R_i$ is at most $\bigO(\gcd(m,k))$.
  Let $t_0$ denote the number of steps during $R_i$ where
  there is one more anchor in $[0,p/2)$ than in $[p/2, p)$, and
  let $t_1$ denote the number of other time steps.
  Clearly $|t_1-t_0|\le \bigO(1)$.
  Then, utilizing the tight concentration of points around the
  anchors from \cref{clm:structureSi} we have that the
  contribution of $R_i$ is at most $\bigO(\gcd(m,k))$.
  Finally, the number of full revolutions is at most $\bigO(m)$,
  so the total contribution from all full revolutions is at most
   $\bigO(m \gcd(m,k))$.
 Summing all the contributions discussed in the proof gives the
 desired bound.
\end{proof}



We now use the powerful combinatorial  \cref{lem:epicbound}
to conclude the proof of \cref{thm:dontneedb}.

\begin{theorem}\label{cor:kindaepicevenifweak}
  $\sum_{x\in X} \lap(x) \le \tilo(p).$
\end{theorem}
\begin{proof}
  In \cref{lem:pigeons} we showed that each $x\in X$ can be
  represented by some $\modp(\sigma m^{-1} k)$ for $\sigma\in \pm
  1, m\in [n], k\in [\ceil{p/n}]$;
  Form a set $Y$ of triples by selecting for each $x\in X$ some
  such representative $(m,k,\sigma)$. Of course each triple can
  only represent one $x$.
  Then by \cref{lem:epicbound} we have:
   \[
     \sum_{x\in X} \lap(x) \le  \sum_{(m,k,\sigma)\in Y} \bigO\left(k+\left(m + \frac{p}{mk}\right)\cdot
     \gcd(k,m)\right) \le \bigO(p) + \bigO(p)\cdot\sum_{(m,k,\sigma)\in Y}
     \frac{\gcd(k,m)}{km} ,
   \]
   where the final inequality was obtained by replacing $k,m$ by
   their maximum possible values and using the trivial bound
   $\gcd(k,m)\le m$.
   We now bound the final term in the sum:
\[
  \sum_{(m,k,\sigma)\in Y} \frac{\gcd(k,m)}{km} \le \sum_{\sigma
  = \pm 1}\sum_{m\in [n]}
  \sum_{d\mid m}\sum_{dk\in [\ceil{p/n}]} \frac{d}{dk m} \le \sum_{m\in
  [n]}\frac{\tau(m)}{m}  \bigO(\log n)\le \bigO(\log^3 n).
\]
Where we have used the well-known bound that $\E_{x\in
[n]}[\tau(x)]\le\bigO(\log n)$ \cite{hardy1979introduction}.
Thus we have the desired bound: $\sum_{x\in X} \lap(x)\le
\tilo(p)$.

\end{proof}

Let random variables $C, M$ denote the number of collisions and
maxload respectively. Let  $\mu = \E[M]$. Let $C_{i,j}$ be $1$ if
$i,j$ collide and $0$ otherwise; our convention is $C_{i,i}=0$.
Using \cref{cor:kindaepicevenifweak} we obtain a bound on $\E[C]$.
\begin{cor}\label{lhs:complicated7}
  $\E[C] \le n^2/4 + \tilo(n^{1/2}).$
\end{cor}
\begin{proof}
  Interpreting \cref{cor:kindaepicevenifweak} probabilistically,
  for any $n$-element set $X\subset [p]$ we have
  \begin{equation}\label{eq:justabovefor1}
    \sum_{y \in X} \E[C_{1,y}] \le n/2 + \tilo(1).
  \end{equation}
  \eqref{eq:justabovefor1} is easily generalized to bound the number of collisions
  between $y\in X$ and any $x\neq 0$.
  In particular, the collisions between $x,X$ are the same as the collisions
  between $1,\modp(x^{-1}X)$ where $x^{-1}$ is the inverse of $x \bmod
    p$. Because \eqref{eq:justabovefor1} holds for arbitrary
  $n$-element sets $X$, we have for any $x\neq 0$
  \[
  \sum_{y \in X} \E[C_{x,y}] \le n/2 + \tilo(1).
  \]

  Now, we are equipped to bound $\E[C]$. Index $X$ as $X=\set{x_1,x_2,\ldots,
      x_n}$. Then
  \begin{equation*}
    \E[C]  = \sum_{i=1}^{n}\sum_{j=1}^{i-1} \E[C_{x_i,x_j}] 
          \le \sum_{i=1}^{n} \paren{i/2 + \tilo(1)}  
          \le n^2/4  + \tilo(n).
  \end{equation*}
\end{proof}

Finally, we complete the proof of \cref{thm:dontneedb} by comparing the expected number
of collisions to the number of collisions induced by the maxload.
\begin{proof}[Proof of \cref{thm:dontneedb}]
  The number of collisions $C$ is determined by the maxload. In
  particular,
    $C = \binom{M}{2}+\binom{n-M}{2}.$
    This is a convex function of $M$.
  Thus, applying Jensen's inequality and comparing with 
  \cref{lhs:complicated7} gives:
  \[ \binom{\mu}{2}+\binom{n-\mu}{2} \le n^2/4 + \tilo(n) .\]
  Solving the quadratic in $\mu$ we find:
  $\mu \le n/2+\tilo(\sqrt{n}).$
\end{proof}

\paragraph{Acknowledgements}
The author thanks William Kuszmaul and Martin Farach-Colton for
proposing this problem and for helpful discussions.
The author also thanks Nathan Sheffield for several technical discussions
about the bound of \cref{thm:dontneedb}.



\bibliography{refs}
\appendix

\section{Proof of \cref{thm:Zm1_3}}
\label{sec:formalpfZm13}

In this section we provide the full proof of \cref{thm:Zm1_3},
translating Knudsen's elegant $\widetilde{\bigO}(n^{1/3})$ bound
\cite{knudsen_linear_2017} from $\FH_p$ to $\SLH_m$.
Although Knudsen's proof only requires relatively small
modifications it is difficult to black-box his proof
because our modifications permeate the whole proof. Thus,
we provide a full proof here.
\begin{theorem}
  Fix $m>n^{2.5}$ with $m\in \poly(n)$.\footnote{This is slightly
  different from the convention in the rest of the paper
(\cref{rmk:assumesize}) that $m>n^{6}$. For elegance of
presentation in applying this theorem, and because it does not
add any complexity to the proof, we prove this theorem with this
smaller universe size.}
\[M_{\SLH}(m,n)\le \widetilde{\bigO}(n^{1/3}).\]
\end{theorem}
\begin{proof}[Proof of \cref{thm:Zm1_3}]
  Fix $n$-element set $X\subset [m]$. We abbreviate $M_{\SLH}(m,X)$ to $M$.
  Fix integer $\alpha < n/4$, and let $\eps = \Pr[M > 4\alpha]$.
  We will show that if $\alpha$ is sufficiently large, in
  particular $\alpha > \Omega(n^{1/3}\log n)$, then $\eps$ is
  very small. Define $\mathcal{A}\subset \Z_m^{\times }$ to be
  the set of \defn{bad} $a$'s, i.e., values of $a$ which make the
  maxload exceed $4\alpha.$ Clearly $|\mathcal{A}| /
  |\Z_m^{\times}| = \eps$.

  Throughout this proof we will write $x^{-1}$ to mean the
  multiplicative inverse of $x$ modulo $m$.

  Define a \defn{pre-bin} to be the pre-image of a bin, i.e., 
  $((m/n)\cdot [k, k+1))\cap [m]$ for some $k\in \Z_{\ge 0}$.
  We call this a pre-bin because if $\modm(xa)$ is in pre-bin
  $k$ then $x$ is hashed to bin $k$ under $a$.
  We say that $a_1,a_2 \in \Z_m^{\times}$ are
\defn{close} if \[\circabs_m(a_1 - a_2) \le
\frac{m}{n\alpha}.\]
  Let $w = \floor{\frac{m}{n\alpha}}$. We call an interval of
  length $w$ a \defn{tiny} interval.
  To bound the number of bad $a$'s we analyze the number of
  \defn{close pairs} (i.e., pairs $a_1,a_2$ which are close). 
  Closeness is a refinement of the
  property of lying within the same pre-bin, which we used in
  \cref{prop:sqrtnZ}. Intuitively, close
  pairs lie within a $(1/\alpha)$-fraction contiguous portion of
  a pre-bin, although technically a bin boundary may split a set
  of close pairs into two pre-bins. In fact, we will count close
  pairs by analyzing groups of $\alpha$ adjacent intervals each
  of size $w$; in other words, each group will be
  a partition of a pre-bin into tiny intervals.
  By convention close pairs are unordered, i.e., we do not count
  $b,c$ and $c,b$ as distinct close pairs.
 
 Now we provide intuition helpful for counting the number of
 elements which are close to some fixed element $b \in (\alpha,
 2\alpha) \cap \Z_m^{\times}$. 
 \begin{claim}\label{clm:prebintiny}
   Assume $b\in \Z_m^{\times}\cap (\alpha, 2\alpha).$\footnote{We will show later that
   this interval is non-empty for the relevant values of
 $\alpha$.}
   Let $I$ be a pre-bin. Then $\modm(b^{-1}\cdot I)$ is
   contained in the union of $b$ tiny intervals. 
 \end{claim}
 \begin{proof}
For each $j\in [w]$ we have
\[b^{-1}\cdot (j+[w]b) \equiv b^{-1}j + [w]  \mod m.\]
In other words, we can partition $x\in I$ based on
$\posmod_b(x)$; $x,y\in I$ for which $x\equiv y \bmod b$ are close
after multiplication by $b^{-1}$.
Note that many of these intervals may be empty; i.e., it is
possible that $\modm(b^{-1}\cdot I)$ is contained in
substantially less than $b$ tiny
intervals, but in general we cannot give a stronger bound.
 \end{proof}
 \begin{claim}\label{clm:oneBclosePairs}
   Let $b\in \Z_m^{\times}\cap (\alpha, 2\alpha).$
Say that for some pre-bin $I_b$ the set $\modm(b\cdot X) \cap
I_b$ contains at least $4\alpha$ elements. 
Then $X$ contains at least $\Omega(\alpha)$ close pairs.
 \end{claim}
 \begin{proof}
By \cref{clm:prebintiny} the $4\alpha$ elements of $
X\cap \modm(b^{-1} I_b)$ are distributed amongst $b\le 2\alpha$ tiny intervals.
Any elements within the same tiny interval constitute a close
pair. The number of close pairs we obtain from splitting these
$4\alpha$ elements amongst $b$ tiny intervals is
minimized if we distribute the $4\alpha$ elements evenly amongst
the tiny intervals. 
However, even if the elements are distributed evenly we still
have at least $2$ elements per interval, and thus $\Omega(\alpha)$
total close pairs.
We remark that if we had defined a tiny interval to be any
smaller then this argument would not guarantee any
close pairs.
 \end{proof}

Now we explore the relationship between close pairs and maxload.
For each bad $a$ there is some pre-bin $I_{a}$ which
contains at least $4\alpha$ elements of  $\modm(a X)$.
By \cref{clm:oneBclosePairs}, this gives us at least $\alpha$
close pairs in $X$. 
Furthermore, we will show in \cref{lem:closepairs} that among
bad $a$'s in a suitably chosen subset $B\subset \mathcal{A}$, the close
pairs given by each $a\in B$ are fairly disjoint.
Intuitively this means that the more bad $a$'s there are, the
more close pairs there are. 
Similarly to in \cref{prop:sqrtnZ}, we will
conclude by counting the close pairs with a different
method to show that too-large maxload results in too many close pairs.

  We proceed to formalize this reasoning.
  Let \[U = \Z_m^\times \cap \PRM \cap (\alpha, 2\alpha).\]
  The multiplier $a\gets \Z_m^{\times }$ is a random variable.
  Define
  \[B= U\cap \modm(a^{-1}\mathcal{A});\] $B$ is a random variable
  dependent on $a$.

  \begin{claim}\label{clm:sizeB}
    $\E[|B|] \geq \Omega\paren{\frac{\alpha}{\log \alpha} - \log
    n}\cdot \eps.$
  \end{claim}
  \begin{proof}
    The prime number theorem says that there are at least 
    $\Omega\left(\alpha/\log \alpha\right)$ primes in the interval
    $(\alpha, 2\alpha)$.
    On the other hand, $m$ cannot have more than $\log m$
    distinct prime divisors. 
    Hence, by excluding the prime divisors of $m$ from $\PRM \cap
    (\alpha, 2\alpha)$  we find:
    \begin{equation}
      \label{eq:ubound}
  |U| \geq \Omega\left(\frac{\alpha}{\log \alpha} - \log n\right).
    \end{equation}

    For any unit, and in particular for any bad $a_0\in
    \mathcal{A}$, $\modm(a^{-1}a_0)$ is uniformly random in
    $\Z_m^{\times}$. Thus,
    \[\Pr[\modm(a^{-1}a_0)\in U]  = |U| / |\Z_m^{\times }|.\]
     There are $|\mathcal{A}|=|\Z_m^{\times }|\eps$ bad $a_0$'s, so by linearity of
    expectation we have
    \[\E[|B|] = |\Z_m^{\times }|\eps |U|/|\Z_m^{\times }| =
    |U|\eps,\]
    which combined with \cref{eq:ubound} gives the desired
    result.
  \end{proof}

  \begin{claim}
    \label{clm:ezclosepairs}
    The expected number of close pairs is at most
    $\bigO(\frac{n}{\alpha}\log\log n)$.
  \end{claim}
  \begin{proof}
    As in the proof of \cref{prop:sqrtnZ}, we define
    \defn{linked} and \defn{unlinked} pairs. 
    We say that distinct $x,y\in X$ are linked if $\gcd(x-y, m)
    > w$ and unlinked otherwise. 
    If $x,y$ are linked, they cannot be close by virtue of being distance at
    least $w$ apart. For unlinked $x,y$ 
    $\floor{\modm(a(x-y))/w}$ would be nearly uniformly
    distributed on
    $[\floor{m/w}]$ if $a$ were chosen randomly from $\Z_m$
    (in particular, no value would be more than a factor-of-$2$
    more probable than other values).
    Using \cref{fact:toitent} we find the probability of $x,y$
    being close when $a\gets \Z_m^{\times}$ is at most 
    $\bigO\paren{\frac{\log\log n}{n\alpha}}.$
    There are $\binom{n}{2}$ total pairs. Then, by linearity of
    expectation there are $\bigO\left(\frac{n}{\alpha}\log\log
      n\right)$
    expected close pairs. 
  \end{proof}

  Now we establish the key combinatorial lemma:
  \begin{lemma}
    \label{lem:closepairs}
    There are at least $|B|\alpha/2$ close pairs.
  \end{lemma}
  \begin{proof}

    Fix $a\in \Z_m^{\times}$. By definition of $B$, each $b\in B$ can be expressed
    as $b=\modm(a^{-1}a_0)$ for some bad $a_0\in
  \mathcal{A}$. By definition of $a_0$ being bad there exists
  a pre-bin $I_{b}$ (which of course has size
  $|I_{b}|\in \set{\floor{m/n}, \ceil{m/n}}$)
  such that at least $4\alpha$ elements from $X$ fall in $I_b$
  under $a_0$. That is,
  \begin{equation}\label{eq:a0isbadandsobisgood}
  |I_{b}\cap \modm(a_{0}X)| = |I_b  \cap \modm(baX)| > 4\alpha.
  \end{equation}
  Recall \cref{clm:prebintiny}, \cref{clm:oneBclosePairs}:
  multiplication by $b^{-1}$ ``fractures'' the interval $I_b$ into
  at most $b$ tiny intervals, and using these tiny intervals we
  can obtain $\Omega(\alpha)$ close pairs. Now we analyze the
  overlap of the close pairs given by different $b\in B$.

  \begin{claim}\label{clm:IBCone}
   For distinct $b,c\in B$ and pre-bins $I_b,I_c$ the set
   \[\modm(b^{-1} I_{b})\cap \modm(c^{-1}I_c)\]
   is contained in a single tiny interval.
  \end{claim}
  \begin{proof}
    Without loss of generality let $b<c$.
    For some $\delta_1,\delta_2\in \Z$ we can
    write\footnote{Technically the size of a pre-bin can be
      either $\floor{m/n}$ or $\ceil{m/n}$. We take the pre-bins
      to be of size $\ceil{m/n}$ in this proof; clearly this can
      only serve to increase the intersection of the pre-bins, so
    this assumption is without loss of generality.}
    \[I_b = [\ceil{m/n}]+\delta_1,I_c=[\ceil{m/n}]+\delta_2.\]
    Thus, to understand $\modm(b^{-1}I_b) \cap \modm(c^{-1}I_c)$
    we may equivalently study 
    \begin{equation}\label{eq:onetinyintervalstudy}
      \paren{\modm(bc^{-1})[\ceil{m/n}]} \cap
      \paren{[\ceil{m/n}]+\delta}
    \end{equation}
    for $\delta\in \Z$.
    The fact that $b,c\in (\alpha, 2\alpha)\cap \PRM \cap
    \Z_m^{\times }$, i.e., are
    co-prime, of similar size, and coprime to $m$, strongly restricts the behavior
   of $\modm(bc^{-1})$. In particular there exists $\lambda \in
   \N$ with
   \[\lambda m + b \equiv 0 \mod c.\]
   For this value of $\lambda$,
   \[\frac{\lambda m+b}{c} \cdot c \equiv b \mod m.\]
    In other words, 
    \[\modm(bc^{-1}) = \frac{\lambda m + b}{c}.\]

    Now, we use this formulation of $\modm(bc^{-1})$ to study
    \cref{eq:onetinyintervalstudy}.
    Because $c\perp m$, there is a permutation $\pi$ of $[c]$ so
    that 
    \begin{equation}\label{eq:permutemup}
    \modm(j \cdot \lambda m / c) = \pi_j m/c.
    \end{equation}
    Because $b<2\alpha$, \cref{eq:permutemup} implies that for $j\in [c]$ 
    \begin{equation}\label{eq:stuffisnice}
    \abs{\modm\left(\frac{\lambda m + b}{c} j\right) - \pi_j
    \cdot\frac{m}{c}} < 2\alpha.
    \end{equation}
    Intuitively, because $m/c$ is much larger $\alpha$,
    \cref{eq:stuffisnice} means that for any $\ell$ the values
    $\modm((bc^{-1})(\ell+[c]))$ are essentially spaced out by
    $m/c$. 
    Formally, by our assumption $m>n^{2.5}$ we have:
    \begin{equation}\label{eq:yupnicenice}
      m/c - 4\alpha \ge 2m/n - n > \ceil{m/n}.
    \end{equation}
    
    \cref{eq:stuffisnice} combined with \cref{eq:yupnicenice}
    is a good start to addressing \cref{eq:onetinyintervalstudy},
    showing that if $c$ consecutive numbers are multiplied by
    $\modm(bc^{-1})$ at most one lies in the interval
    $[\ceil{m/n}]+\delta$.
    To finish we analyze numbers which differ by a multiple of
    $c$ within the same interval.
    For any $k \in [\ceil{m/(nc)}]$ we have
    \begin{equation}\label{eq:globalbxnice}
    \frac{\lambda m+b}{c} ck \equiv kb \mod m.
    \end{equation}
    In particular, $kb \le (m/(nc))\cdot b <\ceil{m/n}$ by our
    assumption $b<c$.
    Thus, if $x,y \in [\ceil{m/n}]$ differ by a multiple of
    $c$ then they lie in the same interval of size $\ceil{m/n}$.
    Combining \cref{eq:globalbxnice} with \cref{eq:stuffisnice}
    and \cref{eq:yupnicenice} there are at most $\ceil{m/n}/c$ points
    in the intersection \cref{eq:onetinyintervalstudy}.

    We have described the shape of $\modm(bc^{-1}[m/n])$: it consists of
    $c$ well-separated concentrated intervals of length
    $\ceil{m/n}$ with at most $\ceil{m/n}/c$ elements per each such
    concentrated interval.
    Thus, upon intersection with $[\ceil{m/n}]$ we obtain at most
    $\ceil{m/n}/c$ points. Let $\rho\le
    \ceil{m/n}/c<\frac{m}{n\alpha}$ be the number of points.
    In particular, by \cref{eq:globalbxnice} there is $\delta'\in
    \Z$ so that the points are of the form 
    \begin{equation}\label{eq:rhorhopoints}
    \delta', \delta' + 1b,\delta'+2b,\ldots, \delta' + \rho b.
    \end{equation}
    Multiplying by $b^{-1}$ to translate the points
    \cref{eq:rhorhopoints} in the
    intersection \cref{eq:onetinyintervalstudy} to the
    intersection from the claim statement we find
    \begin{align*}
      &\modm(b^{-1}I_b) \cap \modm(c^{-1}I_c) \\
      &=\modm(b^{-1}\delta'), \modm(b^{-1}(\delta'+b)), \ldots,
      \modm(b^{-1}(\delta'+b\rho))\\
      &=\modm(b^{-1}\delta'), \modm(b^{-1}\delta')+1, \ldots,
      \modm(b^{-1}\delta')+\rho.
    \end{align*}
    Because $\rho < \frac{m}{n\alpha}$ all these points are
    contained in a single tiny interval, as claimed.

  \end{proof}

  Now we use \cref{clm:IBCone} to show that the close pairs given
  by each $b\in B$ are mostly disjoint.
  \begin{claim}
    There are at least $|B|\alpha/2$ close pairs.
  \end{claim}
  \begin{proof}
    For $b \in B, j\in [b]$ let $I_{b,j}$ denote our partition of
  $\modm(b^{-1}I_b)$ into tiny intervals as described in
  \cref{clm:prebintiny}. In particular, each $I_{b,j}$
  is a tiny interval, and 
  \[\bigsqcup_{j\in [b]} I_{b,j} = \modm(b^{-1}\cdot I_b).\]
  For each $b\in B,j\in [b]$ let $\psi(b,j)$ denote the number of
  $c\in B$ such  that $I_{b,j}\cap \modm(c^{-1}I_{c})\neq
  \varnothing$. Note that $\psi(b,j)\geq 1$ because
  $I_{b,j}\cap \modm(c^{-1}I_c)\neq \varnothing$ for $c=b$.
  On the other hand, for each $b\neq c$ the intersection
  $\modm(b^{-1}I_b)\cap \modm(c^{-1}I_c)$ consists of at most a single tiny
  interval by \cref{clm:IBCone}. Therefore,
  \begin{equation}
    \label{eq:deltasum}
  \sum_{j\in [b]} \psi(b,j) < |B| + b \leq 3\alpha.
  \end{equation}
  Define 
  \[\xi(b,j) = \max(0, |\modm(aX)\cap I_{b,j}|-1).\]
  Recall that any two elements in the same tiny interval are
  close. Thus, the number of close pairs contributed by $I_{b,j}$
  is at least 
  \begin{equation}\label{eq:Ibjaxxi}
  \binom{|I_{b,j}\cap \modm(aX)|}{2}\ge \frac{1}{2}\cdot (\xi(b,j))^2.
  \end{equation}
  Of course the $I_{b,j}$'s are not disjoint, but each close pair
  in  $I_{b,j}$ is counted in at most  $\psi(b,j)$ tiny intervals
  $I_{c,j'}$. Combined with \eqref{eq:Ibjaxxi} this shows that
  the number of close pairs is at least
  \begin{equation}
    \frac{1}{2}\sum_{b\in B}\sum_{j\in [b]}
    \frac{\xi(b,j)^2}{\psi(b,j)}.\label{eq:xixi}
  \end{equation}
  Using the Cauchy-Shwarz Inequality on the inner sum of
  \eqref{eq:xixi} gives: 
  \begin{equation}
    \label{eq:cauchy}
    \sum_{j\in [b]} \frac{\xi(b,j)^2}{\psi(b,j)} \geq
    \frac{\paren{\sum_{j\in[b]}\xi(b,j)}^2}{\sum_{j\in [b]}\psi(b,j)}.
  \end{equation}
  Using \cref{eq:a0isbadandsobisgood} we bound the numerator of
  \eqref{eq:cauchy}:
  \begin{equation}\label{eq:numerbound}
  \paren{\sum_{j\in [b]} \xi(b,j)}^2 \ge \paren{4\alpha - b}^2 \geq
  \paren{2\alpha}^2.
  \end{equation}
  Combining \cref{eq:numerbound} and \cref{eq:deltasum}, which
  bound the numerator and denominator respectively of
  \eqref{eq:cauchy}, we obtain 
  \begin{equation}\label{eq:boundinside}
  \sum_{j\in [b]} \frac{\xi(b,j)^2}{\psi(b,j)} \ge
  \frac{4\alpha^2}{3\alpha}\ge \alpha.
  \end{equation}
  Using \cref{eq:boundinside} in \cref{eq:xixi}, we find
  that the number of close pairs is at least $\alpha
  |B|/2$, as desired.
  \end{proof}
  \end{proof}

  \begin{cor}\label{cor:finisher}
  \[\eps < \bigO\paren{\frac{n\log\log n}{\alpha^2( \alpha / \log
  \alpha - \log n )}}. \]
  \end{cor}
  \begin{proof}
    \cref{lem:closepairs} gives a lower bound on the number of
    close pairs: there are at least $|B|\alpha/2$ close pairs.
    \cref{clm:sizeB} gives a lower bound on $\E[|B|]$. 
    \cref{clm:ezclosepairs} gives an upper bound on the number of
    close pairs. Comparing our upper bound and lower bound gives the
    desired inequality for $\eps.$
  \end{proof}

  Finally, we use \cref{cor:finisher} to conclude the proof.
  \begin{cor}
    \[M_{\SLH}(m,n)\le \widetilde{\bigO}(n^{1/3}).\]
  \end{cor}
  \begin{proof} 
    The expression in \cref{cor:finisher} is slightly
    complicated. For $\alpha > n^{1/3}\log n$ we can perform the
    following simplification:
    \[\frac{1}{\alpha/\log \alpha - \log n} < \bigO\left(\frac{\log
    n}{\alpha}\right)\]
    which is true because
    \[\frac{\log n}{\log \alpha}\alpha - \log^2 n >
    \Omega(\alpha).\]
    Thus, for $k > n^{1/3}\log n$ we have
    \[\Pr[M\ge k]\le \bigO\left( \frac{n \log n\log\log n}{k^3}
    \right).\]
    Now we bound $\E[M]$ as follows:
  \begin{align*}
    \E[M] \leq \sum_{k\ge 0} \Pr[M \geq k] 
        \leq \widetilde{\bigO}(n^{1/3}) + \bigO\paren{\sum_{k>n^{1/3}\log n} \frac{n \log n\log\log n}{k^3}}.
\end{align*}
A basic fact of calculus is that 
\[\sum_{k > n^{1/3}\log n} \frac{1}{k^3} \le
\bigO\paren{\frac{1}{n^{2/3}\log^2 n}}.\]
Thus, our bound for $\E[M]$ simplifies to 
\[\E[M] \le \widetilde{\bigO}(n^{1/3}).\]
This bound was for arbitrary $X$, and thus also holds for
worst-case $X$.

  \end{proof}

\end{proof}

\section{Maxload of a Structured Set}
\label{sec:nice}
Intuitively, because random sets have small maxload any set that
has achieves large maxload, if such a set exists, must be
somehow ``highly structured''.
In this section we investigate the maxload of one natural
candidate for such a structured set. 
However, rather than achieving abnormally large maxload, we prove
that this set exhibits constant maxload.
This provides further evidence to support the hypothesis that
$\LH$ achieves small maxload.
\begin{theorem}
  \label{thm:nisnice}
$\RH_u$ achieves maxload $\Theta(1)$ on $[n]$.
\end{theorem}
  For multiplier $a\in (0,1)$ define
  \[g(a) = \min\setof{k\in \N}{\circabs_1(ak)< 1/n}.\]
  Let random variable $M$ denote the maxload of $\RH_u$ on
  $X=[n]$.

\begin{claim}\label{clm:kbiggerthanNisSilly} 
  $\E[M \mid g(a) \ge n-2] \le \bigO(1).$
\end{claim}
\begin{proof}
  If $i,j\in [n]$ hash to the same bin then $\circabs_1(a(i-j))
  < 1/n$ and so $g(a)\le |i-j|.$ Thus, conditional on $g(a)\ge
  n-2$ any colliding pair $i,j$ must satisfy $|i-j|\ge n-2$.
  There are only $\bigO(1)$ such pairs $\set{i,j}$ namely
  $\set{0,n-1}, \set{0,n-2}, \set{1,n-1}$.
\end{proof}
  
  \begin{claim}\label{clm:proffak}
    Let $k\le n$. Then $\Pr[g(a) = k] \le 2/n.$
  \end{claim}
  \begin{proof}
    If $g(a) \mid k$ there must be $c \in [k]$ such that
    \begin{equation}\label{eq:ckdelta}
  a\in \frac{c}{k} + \frac{1}{k}\cdot [-1/n,1/n].
    \end{equation}
  Note that if $c\not\perp k$ then $g(a) < k.$
  The probability of $a$ lying in
  this union of $\phi(k)$ intervals each of length $2/(nk)$ is
  at most
  \[\phi(k)\cdot \frac{2}{nk} \le 2/n,\]
  which bounds $\Pr[g(a)=k]$.

  \end{proof}

  \begin{lemma}\label{lem:lnnkyay} Let $k\in \N$ with ${1<k<
    n-2}$. Then $\E[M\mid g(a)=k] \le \bigO(\ln (n/k)).$
  \end{lemma}
  \begin{proof}
    Throughout the proof we condition on $g(a)=k$.
    We partition $a\cdot [n]$ into $k$ \defn{clumps} $C_1,\ldots, C_k$
    as follows: for each $i\in [n]$, $a\cdot i \in C_{\posmod_k(i)}.$
    \begin{claim}\label{clm:clumpsnooverlap}
      If $i\not\equiv j \bmod k$ then $i,j$ hash to distinct bins.
    \end{claim}
    \begin{proof}
    As noted in \cref{eq:ckdelta},  we can write $a$ in the form 
       $a = c/k + \delta$ for some $c\perp k$ and $\delta \in
      [\frac{-1}{nk}, \frac{1}{nk}]$. 
      Intuitively, $\delta$ is small so $a$ behaves similarly to
      $c/k$. For $c/k$ we have that $\posmod_1(i\cdot c/k)$ travels on some
      permutation of $[k]\cdot c/k$ as $i$ travels over $[k]$.
      The term $\delta$ introduces some deviation from the
      behavior of $c/k$. We visualize this in
      \cref{fig:perm-png}: the black rectangles start at
      multiples of $1/k$ but then extend a little further to
      account for the deviation introduced by $\delta$. 
      Now we formally show that the deviation introduced by
      $\delta$ is small.

      The \defn{width} of clump $C_i$ is defined as $\max C_i -
      \min C_i$; width is determined by $\delta$. In
      particular, every $k$ steps we take one step within a
      clump, and the step size is $\delta$. In total this means
      that the widths are 
      \[\delta\cdot \frac{n}{k} \le \frac{1}{nk}\cdot \frac{n}{k} \le
      \frac{1}{k^2}.\]
      On the other hand, clumps are separated by a much larger
      quantity: $1/k.$
      In particular, $k<n-2$ by assumption
      so 
      \begin{equation}\label{eq:obviousthing}
      n(k-1)>(k-1)(k+2)=k^2+k-2 \ge k^2.
      \end{equation}
      Rearranging \cref{eq:obviousthing} gives 
      \[\frac{1}{k} - \frac{1}{k^2} \ge \frac{1}{n}.\]
      In other words, elements lying in different clumps cannot
      lie in the same bin, because there is a gap of at least
      $1/n$ between each of the clumps (this justifies why the
      rectangles in \cref{fig:perm-png} are drawn as
      non-overlapping).
      Note that here we have also applied the important fact that 
      clumps all grow in the same direction, which is determined by $\sgn(\delta)$.
\begin{figure}[h]
  \centering
  \includegraphics[width=0.35\textwidth]{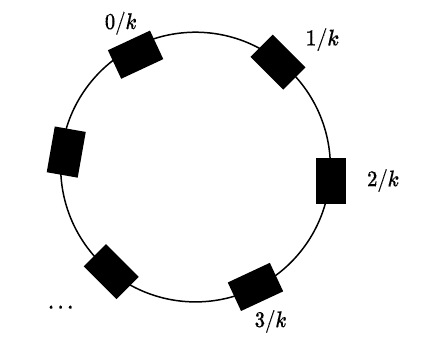}
  \caption{Visualization of clumps.}
  \label{fig:perm-png}
\end{figure}
    \end{proof}

    \begin{claim}\label{clm:maxloadbeps}
      Assume $a = (c+\eps)/k$ for some $\eps \in [-1/n,1/n]$.
      Then the maxload is at most $\floor{(1/n)/\eps}$.
    \end{claim}
    \begin{proof}
      By \cref{clm:clumpsnooverlap} clumps map to separate bins, so
      it suffices to focus on a single clump.
      Observe that for any $i\in [\ceil{n/k}]$,
      $\posmod_1(i\cdot ak) = i\cdot \eps.$
      Thus, it is impossible for more than $\floor{(1/n)/\eps}$ of
      the values in a clump  to lie in the same bin. In other
      words, the maxload is at most $\floor{(1/n)/\eps}$.
    \end{proof}

    Combining \cref{clm:maxloadbeps}, \cref{clm:clumpsnooverlap}
    we compute a bound on the maxload.
    For particularly small $|\eps|$ we use the fact that clumps do
    not intersect to deduce that the maxload is at most $n/k$.
For $\eps$ with $|\eps| > k /n^2$ the bound $(1/n)/\eps$ becomes stronger. 
Clearly each value of $\eps$ is equally likely.
Thus in total we have
\begin{align*}
  \E[M\mid g(a)=k] \le \frac{n}{k}\frac{2k/n^2}{2/n} + 
  \frac{1}{2/n}\cdot 2\int_{k/n^2}^{1/n}(1/n)/\eps \;\d \eps 
    \le \bigO(\ln (n/k)).
\end{align*}
  \end{proof}

  \begin{proof}[Proof of \cref{thm:nisnice}]
  Combining \cref{lem:lnnkyay}, \cref{clm:kbiggerthanNisSilly},
  and \cref{clm:proffak} gives
\begin{align*}
  \E[M] \leq \bigO\left(\sum_{k=2}^{n-2} \frac{\log(n / k)}{n}\right) + \bigO(1) 
        \leq \bigO\left(\frac{\log (n^n/n!)}{n} \right)
        \leq \bigO(1).
\end{align*}
\end{proof}







\section{$\SLH$ on random inputs}
In this section we analyze the behavior of $\SLH_m$ on random
inputs, showing that it is the same as that of $\FH_p$. This is
further evidence of the similarity between $\SLH_m,\FH_p$.
\begin{prop}\label{prop:randominputZ}
  Let $x_1,\ldots, x_n$ be independently randomly selected
  from $\Z_m$. 
  The expected maxload of $\SLH_m$ on $x_1,\ldots, x_m$ is 
  $\bigO\left(\frac{\log n}{\log\log n}\right)$.
\end{prop}
\begin{proof}
  If $x_i$ has $\gcd(x_i,m) < m/n$ then the probability of $x$
  landing in any particular bin is $\bigO(1/n).$

  Fix $x\in [m]$ with $\gcd(x,m)>m/n$. We decompose $x$ into $x =
  yz$ where $y \mid m, z\perp m$, $y >m/n$, and thus $z<n$.
  Because $m$ has at most $m^{o(1)}$ divisors
  (\cref{fact:numdivs}), there can be at most $nm^{o(1)}$ such
  $x.$
  Hence, the probability that $x_i$ and $m$ share such a
  large factor is bounded by $\frac{nm^{o(1)}}{m}$. 
  
  Because $m\ge n^6$ (\cref{rmk:assumesize}), we can take a union bound to deduce
  that with probability at least $1-1/n^{3}$ no $x_i$ shares a
  large common factor with $m$. Hence, the event that some
  $\gcd(x_i, m)$ is large contributes at most $\bigO(1)$ to the
  expected maxload. 
  As noted, if all $x_i$ have $\gcd(x_i, m) < m/n$ then
  each $x_i$ has probability $\bigO(1/n)$ of falling in each
  bin, and the bins that each $x_i$ fall in are independent.
  Then by an almost identical analysis to the standard problem of
  throwing $n$ balls into $n$ bins we clearly achieve the desired maxload
  bound.
\end{proof}

\end{document}